\documentclass[a4paper]{article}
\usepackage[a4paper, total={6in, 8in}]{geometry}
\usepackage[toc,page]{appendix}
\usepackage{multirow}
\usepackage{textcomp}
\usepackage{enumitem}
\usepackage{amsthm}
\usepackage{graphicx}
\usepackage{subcaption}
\usepackage{float}
\usepackage[nocomma]{optidef}
\usepackage{dsfont}
\usepackage{times}
\usepackage{fancyhdr,graphicx,amsmath,amssymb}
\usepackage[ruled,vlined]{algorithm2e}
\include{pythonlisting}
\usepackage{array}
\newcommand{\PreserveBackslash}[1]{\let\temp=\\#1\let\\=\temp}
\newcolumntype{C}[1]{>{\PreserveBackslash\centering}p{#1}}
\newcolumntype{R}[1]{>{\PreserveBackslash\raggedleft}p{#1}}
\newcolumntype{L}[1]{>{\PreserveBackslash\raggedright}p{#1}}
\allowdisplaybreaks
\usepackage{diagbox}
\usepackage{booktabs,xcolor,siunitx}
\usepackage{bm}
\usepackage{cite}
\usepackage{wrapfig}
\usepackage{xcolor}
\usepackage{amsmath}
\usepackage{amssymb}
\usepackage{mathtools}

\DeclareMathOperator*{\argmax}{arg\,max}
\makeatletter
\def\BState{\State\hskip-\ALG@thistlm}
\makeatother
\usepackage{fixltx2e}

\usepackage{stfloats}

\begin{document}

%
\title{\bt{Dynamic Pricing and Fleet Management for Electric Autonomous Mobility on Demand Systems}}

\author{Berkay Turan \quad Ramtin Pedarsani \quad Mahnoosh Alizadeh}

\date{}

%


\newcommand{\bt}[1]{{\color{black}#1}}
\newcommand{\btt}[1]{{\color{black}#1}}

\maketitle
\interfootnotelinepenalty=10000
\begin{abstract}
The proliferation of ride sharing systems is a major drive in the advancement of autonomous and electric vehicle technologies. This paper considers the joint routing, battery charging, and pricing problem faced by a profit-maximizing transportation service provider that operates a fleet of autonomous electric vehicles.  We first establish the static planning problem by considering time-invariant system parameters and determine the optimal static policy. While the static policy provides stability of customer queues waiting for rides even if consider the system dynamics, we see that it is inefficient to utilize a static policy as it can lead to long wait times for customers and low profits. To accommodate for the stochastic nature of trip demands, renewable energy availability, and electricity prices and to further optimally manage the autonomous fleet given the need to generate integer allocations, a real-time policy is required. The optimal real-time policy that executes actions based on full state information of the system is the solution of a complex dynamic program. However, we argue that it is intractable to exactly solve for the optimal policy using exact dynamic programming methods and therefore apply deep reinforcement learning to develop a near-optimal control policy.  The two case studies we conducted in Manhattan and San Francisco demonstrate the efficacy of our real-time policy in terms of network stability and profits, while keeping the queue lengths up to 200 times less than the static policy.
\\\\
\providecommand{\keywords}[1]{\textbf{\textit{Keywords---}} #1}
\\
\end{abstract}

%
\newtheorem{proposition}{Proposition}
\newtheorem{corollary}{Corollary}[proposition]
\newtheorem{theorem}{Theorem}
\newtheorem{lemma}{Lemma}
\makeatletter
\def\blfootnote{\xdef\@thefnmark{}\@footnotetext}
\makeatother
\renewcommand{\thefootnote}{\fnsymbol{footnote}}
 \blfootnote{This work is supported by the NSF Grant 1847096. B. Turan, R. Pedarsani, and M. Alizadeh are with the Department of Electrical and Computer Engineering, University of California, Santa Barbara, CA, 93106 USA e-mail:\{bturan,ramtin,alizadeh\}@ucsb.edu.}
 \renewcommand{\thefootnote}{\arabic{footnote}}
\section{Introduction}
The rapid evolution of enabling technologies for autonomous driving coupled with advancements in eco-friendly electric vehicles (EVs) has facilitated state-of-the-art transportation options for urban mobility. Owing to these developments in automation, it is possible for an autonomous-mobility-on-demand (AMoD) fleet of autonomous EVs to serve the society's transportation needs, with multiple companies now heavily investing in AMoD technology \cite{companies}.

The introduction of autonomous vehicles for mobility on demand services
provides an opportunity for better fleet management. Specifically, idle vehicles can be \textit{rebalanced} throughout the network in order to prevent accumulating at certain locations and to serve induced demand at every location. Autonomous vehicles allow rebalancing to be performed centrally by a platform operator who observes the state of all the vehicles and the demand, rather than locally by individual drivers. Furthermore, EVs provide opportunities for cheap and environment-friendly energy resources (e.g., solar energy). However, electricity supplies and prices differ among the network both geographically and temporally. As such, this diversity can be exploited for cheaper energy options when the fleet is operated by a platform operator that is aware of the electricity prices throughout the whole network. Moreover, a dynamic pricing scheme for rides is essential to maximize profits earned by serving the customers. Coupling an optimal fleet management policy with a dynamic pricing scheme allows the revenues to be maximized while reducing the rebalancing cost and the waiting time of the customers by adjusting the induced demand.

We consider a model that captures the opportunities and challenges of an AMoD fleet of EVs, and consists of complex state and action spaces. In particular, the platform operator has to consider the number of customers waiting to be served at each location  (ride request queue lengths), the electricity prices, traffic conditions, and the states of the EVs (locations, battery energy levels) in order to make decisions. These decisions consist of pricing for rides for every origin-destination (OD) pair and routing/charging decision for every vehicle in the network. Upon taking an action, the state of the network undergoes through a stochastic transition due to the randomness in customer behaviour, electricity prices, and travel times.

We first adopt the common approach of network flow modeling to develop an optimal static pricing, routing, and charging policy that we use as a baseline in this paper. However, flow-based solutions generate fractional flows which can not directly be implemented. Moreover, a static policy executes same actions independent of the network state and is oblivious to the stochastic events that occur in the real setting. Hence, it is not optimal to utilize the static policy in a real dynamic environment. Therefore, a real-time policy that generates integer solutions and acknowledges the network state is required, and can be determined by solving the underlying dynamic program. Due to the continuous and high dimensional state-action spaces however, it is infeasible to develop an optimal real-time policy using exact dynamic programming algorithms. As such, we utilize deep reinforcement learning (RL) to develop a near-optimal policy. Specifically, we show that it is possible to learn a policy via \bt{Proximal Policy Optimization (PPO) \cite{schulman2017ppo}} that increases the total profits generated by jointly managing the fleet of EVs (by making routing and charging decisions) and pricing for the rides. We demonstrate the performance of our policy by using the total profits generated and the queue lengths as metrics.

Our contributions can be summarized as follows:
\begin{enumerate}
    \item We formalize a vehicle and network model that captures the aforementioned characteristics of an AMoD fleet of EVs as well as the stochasticity in demand and electricity prices.
    \item We analyze the static problem, where we consider a time-invariant environment (time-invariant arrivals, electricity prices, etc.) to characterize the family of policies that guarantee stability of the dynamic system, to gain insight towards the actual dynamic problem, and to further provide a baseline for comparison.
    \item We employ deep RL methods to learn a joint pricing, routing and charging policy that effectively stabilizes the queues and increases the profits.
\end{enumerate}

\begin{figure}[t]
    \centering
    \includegraphics[width=.9\textwidth]{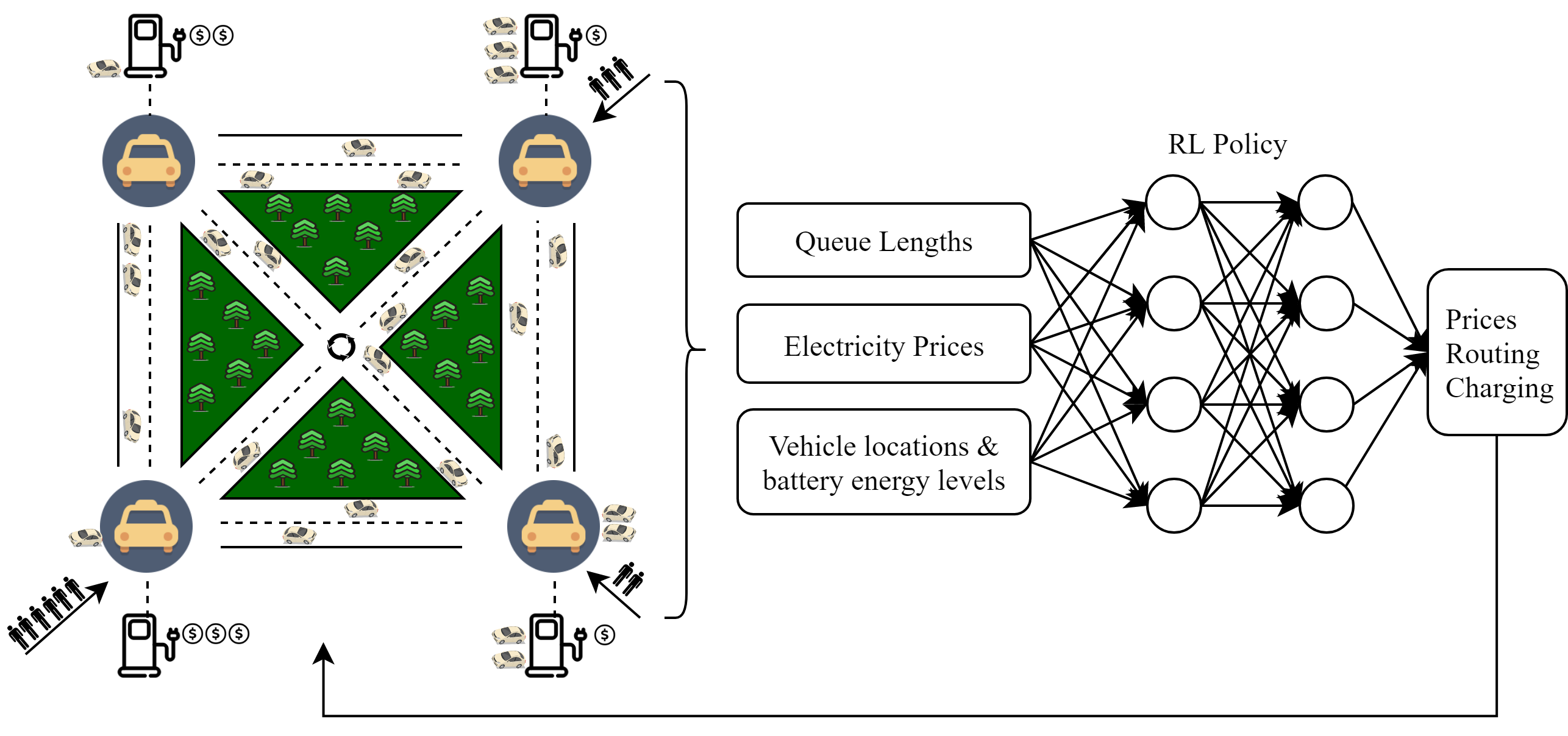}
    \caption{The schematic diagram of our framework. Our deep RL agent processes the state of the vehicles, queues and electricity prices and outputs a control policy for pricing as well as autonomous EVs' routing and charging.}
    \label{fig:schema}
\end{figure}

\begin{wrapfigure}{r}{.3\textwidth}
\vspace{-8em}
    \begin{minipage}{\linewidth}
    \centering\captionsetup[subfigure]{justification=centering}
    \includegraphics[width=\linewidth]{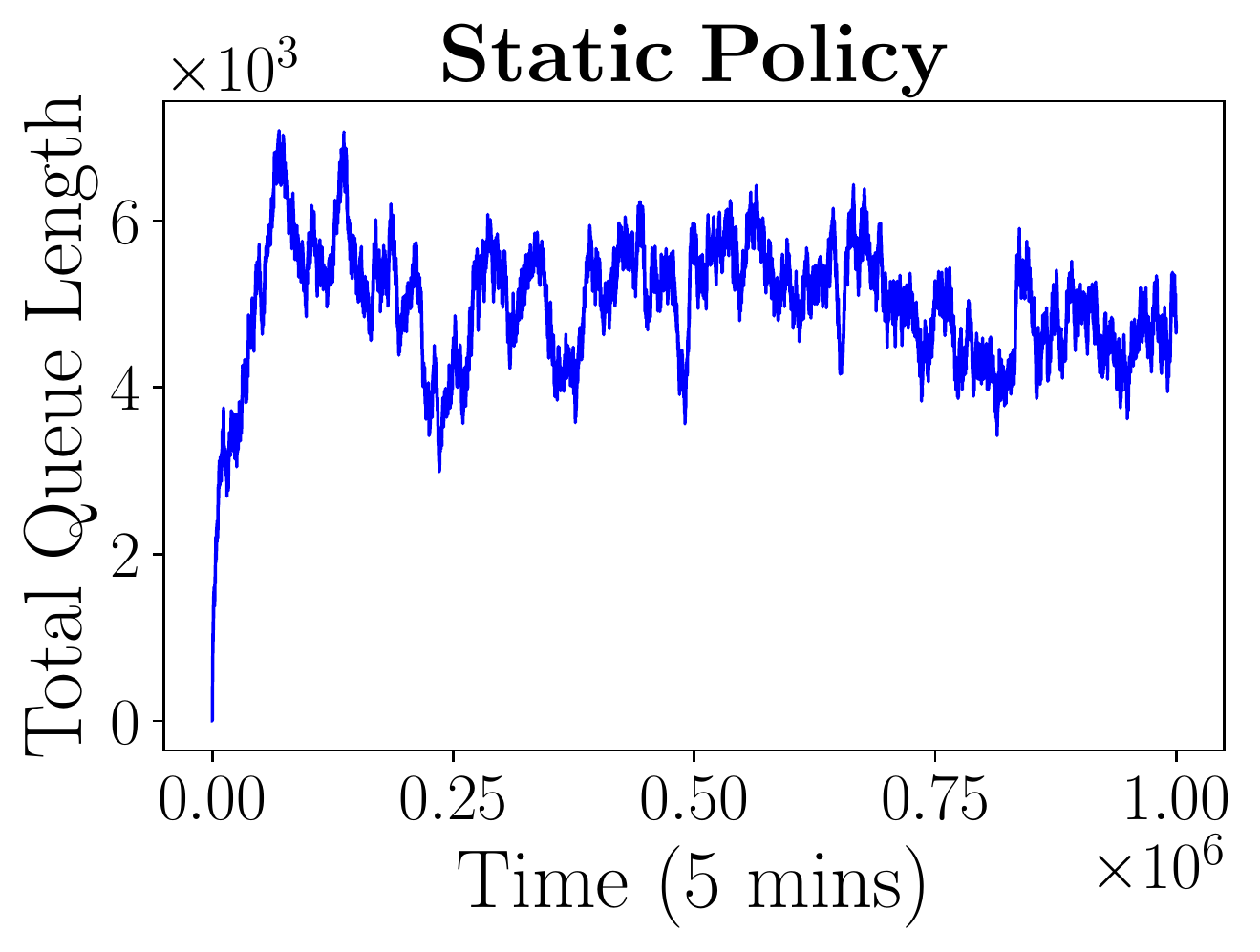}
    \subcaption{}
    \label{fig:prequeuestatic}
    \includegraphics[width=\linewidth]{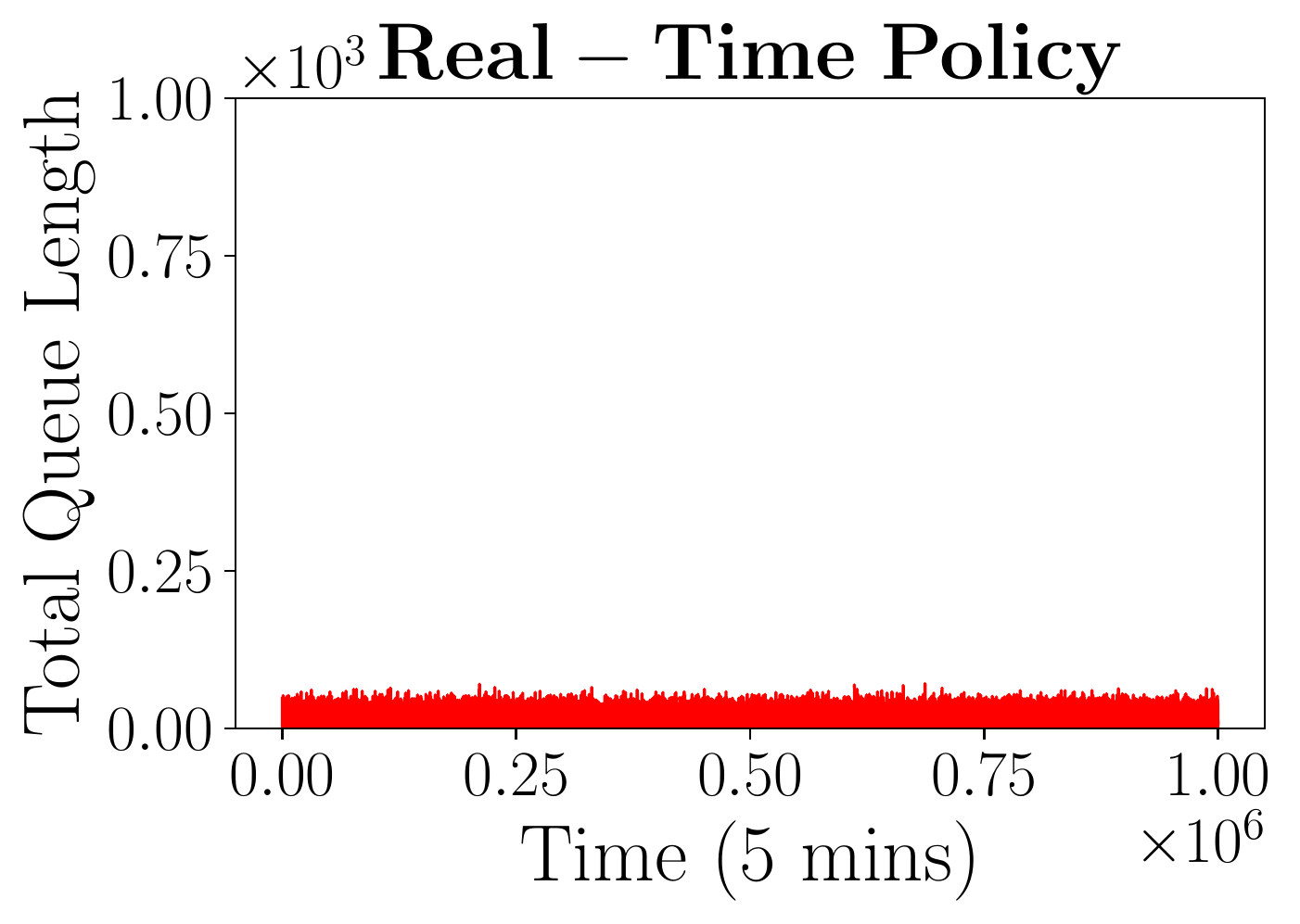}
    \subcaption{}
    \label{fig:prequeuedynamic}
\end{minipage}
\caption{(a) The optimal static policy manages to stabilize the queues over a very long time period but is unable to clear them whereas (b) RL control policy stabilizes the queues and manages to keep them significantly low (note the scales).}\label{fig:preview}
\vspace{-2em}
\end{wrapfigure}

We visualize our real-time framework as a schematic diagram in Figure \ref{fig:schema} and preview our results in Figure \ref{fig:preview}, showing that a real-time pricing and routing policy can successfully keep the queue lengths 400 times lower than the static policy. This policy is also able to decrease the charging costs by $25\%$ by utilizing smart charging strategies (which will be demonstrated in Section~\ref{sec:numerical}).

\vspace{.5em}
\noindent
\textbf{Related work:} Comprehensive research perceiving various aspects of AMoD systems is being conducted in the literature. Studies surrounding fleet management focus on optimal EV charging in order to reduce electricity costs as well as optimal vehicle routing in order to serve the customers and to rebalance the empty vehicles throughout the network so as to reduce the operational costs and the customers' waiting times. Time-invariant control policies adopting queueing theoretical \cite{queuetheoretical}, fluidic \cite{fluidic}, network flow \cite{networkflow}, and Markovian \cite{markovian} models have been developed by using the steady state of the system. The authors of \cite{wei2019ride} consider ride-sharing systems with mixed autonomy. However, the proposed control policies in these papers are not adaptive to the time-varying nature of the future demand. As such, there is work on developing time-varying model predictive control (MPC) algorithms \cite{zhang_rossi_pavone,DBLP:journals/corr/MiaoHLSHZMHP16,datadrivenmpc,mpcmiao,mpcstochastic}. The authors of \cite{datadrivenmpc,mpcmiao} propose data-driven algortihms and the authors of \cite{mpcstochastic} propose a stochastic MPC algorithm focusing on vehicle rebalancing. In \cite{zhang_rossi_pavone}, the authors also consider a fleet of EVs and hence propose an MPC approach that optimizes vehicle routing and scheduling subject to energy constraints. Using a fluid-based optimization framework, the authors of \cite{fluidtimevarying} investigate tradeoffs between fleet size, rebalancing cost, and queueing effects in terms of passenger and vehicle flows under time-varying demand. The authors in \cite{cassandrasarxiv} develop a parametric controller that approximately solves the intractable dynamic program for rebalancing over an infinite-horizon. Similar to AMoD, carsharing systems also require rebalancing in order to operate efficiently. By adopting a Markovian model, the authors of \cite{REPOUX201982} introduce a dynamic proactive rebalancing algorithm for carsharing systems by taking into account an estimate of the future demand using historical data. In \cite{BOYACI2017214}, the authors develop an integrated multi-objective mixed integer linear programming optimization and discrete event simulation framework to optimize vehicle and personnel rebalancing in an electric carsharing system. Using a network-flow based model, the authors of \cite{WARRINGTON2019110} propose a two-stage approximation scheme to establish a real-time rebalancing algorithm for shared mobility systems that accounts for stochasticity in customer demand and journey valuations.

Aside from these, there are  studies on applications of RL methods in transportation such as adaptive routing \cite{MAO2018179}, traffic management \cite{ZHU201430,WALRAVEN2016203}, traffic signal control \cite{ZHU2015487,lisignal2016}, and dynamic routing of autonomous vehicles with the goal of reducing congestion in mixed autonomy traffic networks \cite{lazar2019learning}. Relevant studies to our work aim to develop dynamic policies for rebalancing as well as ride request assignment via decentralized reinforcement learning approaches \cite{rldecentralized1,rldecentralized2,rldecentralized3,linrldecentralized}. In these works however, the policies are developed and applied locally by each autonomous vehicle and this decentralized approach may sacrifice system level optimality. A centralized deep RL approach tackling the rebalancing problem is proposed in \cite{MAO2020102626}, which is closest to the approach we adopt in this paper. Although their study adopts a centralized deep RL approach similar to our paper, they have a different system model and solely focus on the rebalancing problem and do not consider pricing for rides as a control variable for the queues nor the charging problem of EVs as reviewed next.

Regarding charging strategies for large populations of EVs, \cite{New_TSG_SmartEVGrid,New_Survey,New_20} provide in-depth reviews and studies of smart charging technologies. An agent-based model to simulate the operations of an AMoD fleet of EVs under various vehicle and infrastructure scenarios has been examined in \cite{agentbased}. By augmenting optimal battery management of autonomous electric vehicles to the classic dial-a-ride problem (DARP), the authors of \cite{BONGIOVANNI2019436} introduce the electric autonomous DARP that aims to minimize the total travel time of all the vehicles and riders. The authors of \cite{nate} propose an online charge scheduling algorithm for EVs providing AMoD services. By adopting a static network flow model in \cite{berkay}, the benefits of smart charging have been investigated and approximate closed form expressions that highlight the trade-off between operational costs and charging costs have been derived. Furthermore, \cite{rossi_iglesias_alizadeh} studies interactions between AMoD systems and the power grid. In addition, \cite{chen_kockelman} studies the implications of pricing schemes on an AMoD fleet of EVs. In \cite{pricing2}, the authors propose a dynamic joint pricing and routing strategy for non-electric shared mobility on demand services. \cite{jointopt} studies a quadratic programming problem in order to jointly optimize vehicle dispatching, charge scheduling, and charging infrastructure, while the demand is defined exogenously.

To the best of our knowledge, there is no existing work on centralized real-time management for electric AMoD systems addressing the joint optimization scheme of vehicle routing and charging as well as pricing for the rides. In this paper we aim to highlight the benefits of a real-time controller that jointly: (i) routes the vehicles throughout the network in order to serve the demand for rides as well as to relocate the empty vehicles for further use, (ii) executes smart charging strategies by exploiting the diversity in the electricity prices (both geographically and temporally) in order to minimize charging costs, and (iii) adjusts the demand for rides by setting prices in order to stabilize the system (i.e., the queues of customers waiting for rides) while maximizing  profits.

\vspace{.5em}
\noindent
\textbf{Paper Organization:} The remainder of the paper is organized as follows. In Section \ref{systemmodel}, we present the system model and define the platform operator's optimization problem.  In Section \ref{sec:static}, we discuss the static planning problem associated with the system model and characterize the optimal static policy. In Section \ref{sec:dynamic}, we propose a method for developing a near-optimal real-time policy using deep reinforcement learning. In Section \ref{sec:numerical}, we present the numerical results of the case studies we have conducted in Manhattan and San Francisco to demonstrate the performance of our real-time control policy. Finally, we conclude the paper in Section \ref{sec:conclusion}.

\section{System Model and Problem Definition}\label{systemmodel}

\noindent
\textbf{Network and Demand Models:}
We consider a fleet of AMoD EVs operating within a transportation network characterized by a fully connected graph consisting of ${\mathcal M}=\{1,\ldots,m\}$ nodes that can each serve as a trip origin or destination.  
We study a discrete-time system with time periods normalized to integral units $t\in\{0,1,2,\dots\}$. In this discrete-time system, we model the arrival of the potential riders with OD pair $(i,j)$ as a Poisson process with an arrival rate of $\lambda_{ij}(t)$ in period $t$, where $\lambda_{ii}(t)=0$. We adopted a price-responsive rider model studied in \cite{bimpikis}. We assume that the riders are heterogeneous in terms of their willingness to pay. In particular, if the price for receiving a ride from node $i$ to node $j$ in period $t$ is set to $\ell_{ij}(t)$, the induced arrival rate for rides from $i$ to $j$ is given by $\Lambda_{ij}(t)=\lambda_{ij}(t)(1-F(\ell_{ij}(t)))$, where $F(\cdot)$ is the cumulative distribution of riders' willingness to pay with a support of $[0,\ell_{\max}]$\footnote{\bt{For brevity of notation, we uniformly set $\ell_{\max}$ to be the maximum willingness to pay for all OD pairs without loss of generality. Our results can be derived in a similar fashion by replacing $\ell_{\max}$ with $\ell_{\max}^{ij}$, where $\ell_{\max}^{ij}$ is the maximum willingness to pay for OD pair $(i,j)$.}}. Thus, the number of new ride requests in time period $t$ is $A_{ij}(t) \sim \text{Pois}(\Lambda_{ij}(t))$ for OD pair $(i,j)$.

\vspace{.5em}
\noindent
\textbf{Vehicle Model:} 
To capture the effect of trip demand and the associated charging and routing (routing also implies rebalancing of the empty vehicles) decisions on the costs associated with operating the fleet (maintenance, mileage, etc.), we assume that each autonomous vehicle in the fleet has a per period operational cost of $\beta$. Furthermore, as the vehicles are electric, they have to sustain charge in order to operate. Without loss of generality, we assume there is a charging station placed at each node $i \in \mathcal M$. To charge at node $i$ during time period $t$, the operator pays a price of electricity $p_i(t)$ per unit of energy.  
We assume that all EVs in the fleet have a battery capacity denoted as $v_{\max}\in \mathbb Z^+$; therefore, each EV has a discrete battery energy level $v \in \mathcal V$, where $\mathcal V = \{v\in \mathbb{N}| 0\leq v \leq v_{\max}\}$. In our discrete-time model, we assume each vehicle takes one period to charge one unit of energy and $\tau_{ij}$ periods to travel between OD pair $(i,j)$, while consuming $v_{ij}$ units of energy\footnote{\bt{In this paper, we consider the travel times to be constant and exogenously defined for the time period the policy is developed for. This is because we assume that the number of AMoD vehicles is much less compared to the rest of the traffic. Also, to consider changing traffic conditions throughout the day, it is possible to train multiple static and real-time control policies for the different time intervals.}}.

\vspace{.5em}
\noindent
\textbf{Ride \bt{Hailing} Model:}
The platform operator dynamically routes the fleet of EVs in order to serve the demand at each node. Customers that purchase a ride are not immediately matched with a ride, but enter the queue for OD pair $(i,j)$. After the platform operator executes routing decisions for the fleet, the customers in the queue for OD pair $(i,j)$ are matched with rides and served in a first-come, first-served discipline. A measure of the expected wait time is not available to each arriving customer. However, the operator knows that longer wait times will negatively affect their business and hence seeks to minimize the total wait time experienced by users. Denote the queue length for OD pair $(i,j)$ by $q_{ij}(t)$. If after serving the customers, the queue length $q_{ij}(t) > 0$, the platform operator is penalized by a fixed cost of $w$ per person at the queue to account for the value of time of the customers.

\vspace{.5em}
\noindent
\textbf{Platform Operator's Problem:} We consider a profit-maximizing AMoD operator that manages a  fleet of EVs that make trips to provide transportation services to customers. The operator's goal is to maximize profits by 1) setting prices for rides and hence managing customer demand at each node; 2) optimally operating the AMoD fleet (i.e., charging and routing) to minimize operational and charging costs. We will study two types of control policies the platform operator utilizes: 1) a static policy, where the pricing, routing and charging decisions are time invariant and independent of the state of the system; 2) a real-time policy, where the pricing, routing and charging decisions are dependent on the system state.

\section{Analysis of the Static Problem}\label{sec:static}
In this section, we establish and discuss the static planning problem to provide a measure for comparison and demonstrate the efficacy of the real-time policy (which will be discussed in Section~\ref{sec:dynamic}). To do so, we consider the fluid scaling of the dynamic network and characterize the static problem via a network flow formulation. Under this setting, we use the expected values of the variables (arrivals and prices of electricity) and ignore their time dependent dynamics, while allowing the vehicle routing decisions to be flows (real numbers) rather than integers. The static problem is convenient for  determining the optimal static pricing, routing, and charging policy, under which the queueing network of the dynamic system is stable \cite{pedarsani2017robust}\footnote{The stability condition that we are interested in is rate stability of all queues. A queue for OD pair $(i,j)$ is rate stable if $\underset{t\rightarrow \infty}{\lim}q_{ij}(t)/t=0$.}.

\subsection{Static Profit Maximization Problem}

We formulate the static optimization problem via a network flow model that aims to maximize the platform operator's profits.
The platform operator maximizes its profits by setting prices and making routing and charging decisions such that the system remains stable. 

Let $\ell_{ij}$ be the prices for rides for OD pair $(i,j)$, $x_{ij}^v$ be the number of vehicles at node $i$ with energy level $v$ being routed to node $j$, and $x_{ic}^v$ be the number of vehicles charging at node $i$ starting with energy level $v$. We state the platform operator's profit maximization problem as follows:

\begin{subequations}
\label{eq:flowoptimization}
\begin{align}&\underset{x_{ic}^v,x_{ij}^v,\ell_{ij}}{\text{max}}
\label{eq:staticobjective}& &\sum_{i\in{\cal M}}\sum_{j\in{\cal M}}\lambda_{ij}\ell_{ij}(1-F(\ell_{ij}))-\sum_{i\in{\cal M}}\sum_{v=0}^{v_{\max}-1} (\beta+p_i) x_{ic}^v
-\beta \sum_{i\in{\cal M}}\sum_{j\in{\cal M}}\sum_{v=v_{ij}}^{v_{\max}} x_{ij}^v\tau_{ij} \\
& \text{subject to}
\label{eq:staticconst1}& &\lambda_{ij}(1-F(\ell_{ij})) \leq \sum_{v=v_{ij}}^{v_{\max}}x_{ij}^v \quad\forall i,j \in \mathcal M,\\
\label{eq:staticconst2}& & & x_{ic}^v+\sum_{j\in{\cal M}} x_{ij}^v=x_{ic}^{v-1}+\sum_{j\in{\cal M}} x_{ji}^{v+v_{ji}}\quad\forall i\in\mathcal M,\;\forall v\in\mathcal V,\\
\label{eq:staticconst5}& & & x_{ic}^{v_{\max}}=0\quad \forall i\in \mathcal M,\\
\label{eq:staticconst6}& & & x_{ij}^v=0\quad \forall v<v_{ij},\; \forall  i,j\in \mathcal M,\\
\label{eq:staticconst7}& & & x_{ic}^v\geq 0, \; x_{ij}^v\geq 0\; ~\forall i,j\in \mathcal M,\;\forall v\in\mathcal V,\\
\label{eq:staticconst8}& & & x_{ic}^v=x_{ij}^v=0\quad \forall v\notin\mathcal V,\; \forall  i,j\in \mathcal M.
\end{align}
\end{subequations}
The first term in the objective function in \eqref{eq:flowoptimization} accounts for the aggregate revenue the platform generates by providing rides for $\lambda_{ij}(1-F(\ell_{ij}))$ number of riders with a price of $\ell_{ij}$. The second term is the operational and charging costs incurred by the charging vehicles (assuming that $p_i(t)=p_i\;\forall t$ under the static setting), and the last term is the operational costs of the trip-making vehicles (including rebalancing trips).

The constraint \eqref{eq:staticconst1} requires the platform to operate at least as many vehicles to serve all the induced demand between any two nodes $i$ and $j$ (The rest are the vehicles travelling without passengers, i.e., rebalancing vehicles). We will refer to this as the {\it demand satisfaction constraint}. The constraint \eqref{eq:staticconst2} is the  \textit{flow balance constraint} for each node and each battery energy level, which restricts the number of available vehicles at node $i$ and energy level $v$ to be the sum of arrivals from all nodes (including idle vehicles) and vehicles that are charging with energy level $v-1$. 
The constraint \eqref{eq:staticconst5} ensures that the vehicles with full battery do not charge further, and the constraint \eqref{eq:staticconst6} ensures the vehicles sustain enough charge to travel between OD pair $(i,j)$.

The solution to the optimization problem in \eqref{eq:flowoptimization} is the optimal static policy that consists of optimal prices as well as optimal vehicle routing and charging decisions. This policy can not directly be implemented in a real environment because it does not yield integer valued solutions. It is possible generate integer-valued solutions to be implemented in a real environment using the fractional flows (e.g., randomizing the vehicle decisions according to the flows, which we do in Section~\ref{sec:numerical}), yet the methodology is not the focus of our work. Instead, we highlight a sufficient condition for a realizable policy (generating integer valued actions) to provide stability according to the feasible solutions of \eqref{eq:flowoptimization}:


\begin{proposition}\label{prop:stability}

Let $\{\tilde{\ell}_{ij}, \tilde{x}_{ij}^{v}, \tilde{x}_{ic}^{v}\}$ be a feasible solution of \eqref{eq:flowoptimization}. Let $\mu$ be a policy that generates integer actions and can be implemented in the real environment. Then, $\mu$ guarantees stability of the system if for all OD pairs $(i,j)$: 
\begin{enumerate}
    \item The time average of the induced arrivals equals $(1-F(\tilde{\ell}_{ij}))$, and
    \item The time average of the routed vehicles equals $\sum_{v=v_{ij}}^{v_{\max}}\tilde{x}_{ij}^{v}$.
\end{enumerate}
\end{proposition}
The proof of Proposition~\ref{prop:stability} is provided in Appendix~\ref{sec:proofpropstability}. According to Proposition~\ref{prop:stability}, for a static pricing policy with the optimal prices $\ell_{ij}^*$, there exists an integer-valued routing and charging policy that maintains stability of the system.

\begin{corollary}
An example policy that generates integer-valued actions is randomizing according to the flows. Precisely, given a feasible solution $\{\tilde{\ell}_{ij}, \tilde{x}_{ij}^{v}, \tilde{x}_{ic}^{v}\}$ of \eqref{eq:flowoptimization}, integer-valued actions can be generated by routing a vehicle at node $i$ with energy level $v$ to node $j$ with probability
$$\psi_{ij}^v=\frac{\tilde{x}_{ij}^{v}}{\sum_{k=1}^m\tilde{x}_{ik}^{v}+x_{ic}^v},$$
and charging with probability
$$\psi_{ic}^v=\frac{\tilde{x}_{ic}^{v}}{\sum_{k=1}^m\tilde{x}_{ik}^{v}+x_{ic}^v},$$
$\forall i,j\in \cal M$ and $\forall v\in \cal V$. Combining this randomized policy with a static pricing policy of $\ell_{ij}(t)=\tilde{\ell}_{ij},\;\forall t$, results in a policy satisfying the criteria in Proposition~\ref{prop:stability}.
\end{corollary}

The optimization problem in \eqref{eq:flowoptimization} is non-convex for a general $F(\cdot)$. Nonetheless, when the platform's profits are convex in the induced demand $\lambda_{ij}(1-F(\cdot))$, it can be rewritten as a convex optimization problem and can be solved exactly. Hence, we assume that the rider's willingness to pay is uniformly distributed in $[0,\ell_{\max}]$, i.e., $F(\ell_{ij})=\frac{\ell_{ij}}{\ell_{\max}}$\footnote{\bt{It is also possible to use other distributions that might reflect real willingness-to-pay distributions more accurately (such as pareto distribution, exponential distribution, triangular distribution, constant elasticity distribution, and normal distribution). Among these, pareto, exponential, and constant elasticity distributions preserve convexity and therefore the static planning problem can be solved efficiently. Triangular and normal distributions are not convex in their support and therefore the static planning problem is not a convex optimization problem. Nevertheless, it can still be solved numerically for the optimal static policy. Using these distributions however we cannot derive the closed-form results that allow us to interpret the pricing policy of the platform operator. The real-time policy proposed in Section~\ref{sec:dynamic} uses model-free Reinforcement Learning and therefore can be applied using other distributions or any other customer price response model.}}.

\vspace{.5em}
\noindent
\textbf{Marginal Pricing:} The prices for rides  are a crucial component of the profits generated. The next proposition highlights how the optimal prices $\ell_{ij}^*$ for rides are related to the network parameters, prices of electricity, and the operational costs.  
\begin{proposition}\label{prop:marginalprices}
Let $\nu_{ij}^*$ be optimal the dual variable corresponding to the demand satisfaction constraint for OD pair $(i,j)$. The optimal prices $\ell_{ij}^*$ are:
\begin{equation}
\label{eq:optimalprices}
    \ell_{ij}^*=\frac{\ell_{\max}+\nu_{ij}^*}{2}.
\end{equation}
These prices can be upper bounded by:
\begin{equation}\label{eq:boundprices}
         \ell_{ij}^*\leq\frac{\ell_{\max}+\beta(\tau_{ij}+\tau_{ji}+v_{ij}+v_{ji})+v_{ij}p_j+v_{ji}p_i}{2}
\end{equation}
Moreover, with these optimal prices $\ell_{ij}^*$, the profits generated per period is:
\begin{equation}
    \label{eq:profits}
    P=\sum_{i=1}^m\sum_{j=1}^m\frac{\lambda_{ij}}{\ell_{\max}}(\ell_{\max}-\ell_{ij}^*)^2.
\end{equation}
\end{proposition}
The proof of Proposition~\ref{prop:marginalprices} is provided in Appendix~\ref{sec:proofprop2}. Observe that the profits in Equation \eqref{eq:profits} are decreasing as the prices for rides increase. Thus expensive rides generate less profits compared to the cheaper rides and it is more beneficial if the optimal dual variables $\nu_{ij}^*$ are small and prices are close to $\ell_{\max}/2$. We can interpret the dual variables $\nu_{ij}^*$ as the cost of providing a single ride between $i$ and $j$ to the platform. 
In the worst case scenario, every single requested ride from node $i$ requires rebalancing and charging both at the origin and the destination. Hence the upper bound on \eqref{eq:boundprices} includes the operational costs of passenger-carrying, rebalancing and charging vehicles (both at the origin and the destination); and the energy costs of both passenger-carrying and rebalancing trips multiplied by the price of electricity at the trip destinations.
 Similar to the taxes applied on products, whose burden is shared among the supplier and the customer; the costs associated with rides are shared among the platform operator and the riders (which is why the price paid by the riders include half of the cost of the ride).

Although the static policy guarantees stability (by appropriate implementation of integer-valued actions as dictated by Proposition~\ref{prop:stability}), it does not perform well in a real dynamic setting because it does not acknowledge  the stochastic dynamics of the system. On the other hand, a real-time policy that executes decisions based on the current state of the environment would likely perform better (e.g., if the queue length for OD pair $(i,j)$ is very large, then it is probably better for the platform operator to set higher prices to prevent the queue from growing further). Accordingly, we present a practical way of implementing a real-time policy in the next section.




\section{The Real-Time Policy}\label{sec:dynamic}

The static policy established in the previous section has three major issues:
\begin{enumerate}
    \item Because it is based on a flow model, it generates static fractional flows that are not directly implementable in the real setting.
    \item It neglects the stochastic events that occur in the dynamic setting (e.g., the induced arrivals), and assumes everything is deterministic. Hence, it does consider the unexpected occurrences (e.g., queues might build in the dynamic setting, whereas the static model assumes no queues) when executing actions.
    \item It assumes perfect knowledge of the network parameters (arrivals, trip durations, energy consumptions of the trips, and prices of electricity).
\end{enumerate}

Due to the above reasons, it is impractical to implement the static policy in the dynamic environment. A real-time policy that generates integer solutions and  takes into account the current state of the network which is essential for decision making  is necessary, and can be determined by solving the dynamic program that describes the system (with full knowledge of the network parameters) for the optimal policy. Such solutions would address issues 1 and 2 outlined above. Inspired by our theoretical model, the state information that describes the network fully consists of the vehicle states (locations, energy levels), queue lengths for each OD pair, and electricity prices at each node. Upon obtaining the full state information, the actions have to be executed for pricing for rides and fleet management (vehicle routing and charging). Consequent to taking actions, the platform operator observes a reward (consisting of revenue gained by arrivals, queue costs, and operational and charging costs), and the network transitions into a new state (Although the transition into the new state is stochastic, the random processes that govern this stochastic transition is known if the network parameters are known). The solution of this dynamic program is the optimal policy that determines which action to take for each state the system is in, and can nominally be derived using classical exact dynamic programming algorithms (e.g., value iteration). However, the complexity and the scale of our dynamic problem presents a difficulty here: Aside from having a large dimensional state space (for instance, $m=10,\;v_{\max}=5,\;\tau_{ij}=3\;\forall i,j$:\; the state has dimension 1240) and action space, the cardinality of these spaces are not finite (queues can grow unbounded, prices are continuous). \bt{Considering that the computational complexity per iteration for value iteration is ${\cal O}(|{\cal A}||{\cal S}|^2)$ and for policy iteration $\mathcal{O}(|{\cal A}||{\cal S}|^2+|{\cal S}|^3)$ \cite{kaelbling1996reinforcement}, \btt{where $\cal S$ and $\cal A$ are the state space and the action space, respectively,} the problem is computationally intractable to solve using classical dynamic programming. Even if we did make them finite by putting a cap on the queue lengths and discretizing the prices,} curse of dimensionality renders the problem intractable to solve with classical exact dynamic programming algorithms. As such, we resort to approximate dynamic programming methods. Specifically, we define the policy via a deep neural network that takes the full state information of the network as input and outputs the best action\footnote{In general, the policy is a stochastic policy and determines the probabilities of taking the actions rather than deterministically producing an action.}. Subsequently, we apply a model-free reinforcement learning algorithm to train the neural network in order to improve the performance of the policy. Since it is model-free, it does not require a modeling of the network (hence, it does not require knowledge of the network parameters), which resolves the third issue associated with the static policy. 

We adopted a practical policy gradient method, called \bt{Proximal Policy Optimization (PPO)}, developed in \cite{schulman2017ppo}, which is 
effective for optimizing large nonlinear policies such as neural networks. We chose \bt{PPO} mainly because it supports continuous state-action spaces and guarantees monotonic improvement.\footnote{\bt{Although the policy outputs a continuous set of actions, integer actions can be generated by randomizing. This is done during both training and testing, therefore the RL agent observes the integer state transitions and learns as if the policy outputs integer actions. We discuss how to generate integer actions in more detail in Section~\ref{sec:mdp}}.}

We note that it is possible to apply reinforcement learning to learn a policy in any environment, real or artificial, as long as there is data available. In this work we use our theoretical model described in Section~\ref{systemmodel} to create the environment and generate data, mainly because there is no electric AMoD microsimulation environment available and also to verify our findings about the static policy. Developing a microsimulator for electric AMoD (like SUMO \cite{SUMO2018}) and integrating it with a deep reinforcement learning library to create a framework for real traffic experiments remains a future work. \bt{To ensure that our numerical experiments are reproducible, in the next subsection, we describe the Markov Decision Process (MDP) that governs this dynamic environment, which is a direct extension of our static model. It is also possible to enrich the environment and the MDP to reflect real life constraints more accurately such as road capacity and charging station constraints. Since the approach we adopt to develop the real-time policy is model-free, it can be applied identically.}   

In Section~\ref{sec:numerical} we present numerical results on   real-time policies developed through reinforcement learning based on dynamic environments generated through our theoretical model. The goal of the experiments is to primarily answer the following questions:

\begin{enumerate}
    \item Can we develop a real-time control  and pricing policy for AMoD using reinforcement learning and what are its potential benefits over the static policy?
    \item How does the policy trained for a specific network perform, if the network parameters change?
    \item Can we develop a global policy that can be utilized in any network with moderate fine tuning?
\end{enumerate}

\bt{The reader may skip reading Section \ref{sec:mdp} if they are not interested in the details of the MDP model used in our numerical experiment.}

\bt{\subsection{The Real-Time Problem as MDP}\label{sec:mdp}
We define the MDP by the tuple $(\mathcal{S}, \mathcal{A}, \mathcal{T}, r)$, where $\mathcal{S}$ is the state space, $\mathcal{A}$ is the action space, $\mathcal{T}$ is the state transition operator and $r$ is the reward function. We describe these elements as follows:
\begin{enumerate}[wide, labelwidth=!]
    \item $\mathcal{S}$: The state space consists of prices of electricity at each node, the queue lengths for each origin-destination pair, and the number of vehicles at each node and each energy level. However, since travelling from node $i$ to node $j$ takes $\tau_{ij}$ periods of time, we need to define intermediate nodes. As such, we define $\tau_{ij}-1$ number of intermediate nodes between each origin and destination pair, for each battery energy level $v$. Hence, the state space consists of $s_d=m^2+(v_{\max}+1)((\sum_{i=1}^m\sum_{j=1}^m \tau_{ij})-m^2+2m)$ dimensional vectors in  $\mathbb{R}_{\geq 0}^{s_d}$ (We include all the non-negative valued vectors, however, only $m^2-m$ entries can grow to infinity because they are queue lengths, and the rest are always upper bounded by fleet size or maximum price of electricity). As such, we define the elements of the state vector at time $t$ as $\boldsymbol{s}(t)=[\boldsymbol{p}(t)\;\boldsymbol{q}(t)\;\boldsymbol{s_{veh}}(t)]$, where $\boldsymbol{p}(t)=[p_i(t)]_{i\in\mathcal{M}}$ is the electricity prices state vector, $\boldsymbol{q}(t)=[q_{ij}(t)]_{i,j\in\mathcal{M};i\neq j}$ is the queue lengths state vector, and $\boldsymbol{s_{veh}}(t)=[s_{ijk}^v(t)]_{\forall i,j,k,v}$ is the vehicle state vector, where $s_{ijk}^v(t)$ is the number of vehicles at vehicle state $(i,j,k,v)$. The vehicle state $(i,j,k,v)$ specifies the location of a vehicle that is travelling between OD pair $(i,j)$ as the $k$'th intermediate node between nodes $i$ and $j$, and specifies the battery energy level of a vehicle as $v$ (The states of the vehicles at the nodes $i \in \mathcal{M}$ with energy level $v$ is denoted by $(i,i,0,v)$).
    \item $\mathcal{A}$: The action space consists of prices for rides at each origin-destination pair and routing/charging decisions for vehicles at nodes $i\in\mathcal{M}$ at each energy level $v$. The price actions are continuous in range $[0,\ell_{\max}]$. Each vehicle at state $(i,i,0,v)$ ($\forall i \in \mathcal{M},\; \forall v \in \mathcal{V}$) can either charge, stay idle or travel to one of the remaining $m-1$ nodes. To allow for different transitions for vehicles at the same state (some might charge, some might travel to another node), we define the action taken at time $t$ for vehicles at state $(i,i,0,v)$ as an $m+1$ dimensional probability vector with entries in $[0,1]$ that sum up to 1: $\boldsymbol{\alpha_i^v}(t)=[\alpha_{i1}^v(t)\dots \alpha_{im}^v(t)\; \alpha_{ic}^v(t)]$, where $\alpha_{ic}^{v_{\max}}(t)=0$ and $\alpha_{ij}^v(t)=0$ if $v<v_{ij}$. The action space is then all the vectors $\boldsymbol{a}$ of dimension $a_d=m^2-m+(v_{\max}+1)(m^2+m)$, whose first $m^2-m$ entries are the prices and the rest are the probability vectors satisfying the aforementioned properties. As such, we define the elements of the action vector at time $t$ as $\boldsymbol{a}(t)=[\boldsymbol{\ell}(t)\;\boldsymbol{\alpha}(t)]$, where $\boldsymbol{\ell}(t)=[\ell_{ij}]_{i,j\in\mathcal{M},i\neq j}$ is the vector of prices and $\boldsymbol{\alpha}(t)=[\boldsymbol{\alpha_i^v}(t)]_{\forall i,v}$ is the vector of routing/charging actions.
    \item $\mathcal{T}$: The transition operator is defined as $\mathcal{T}_{ijk}=Pr(\boldsymbol{s}(t+1)=j|\boldsymbol{s}(t)=i,\; \boldsymbol{a}(t)=k)$. We can define the transition probabilities for electricity prices $\boldsymbol{p}(t+1)$, queue lengths $\boldsymbol{q}(t+1)$, and vehicle states $\boldsymbol{s_{veh}}(t+1)$ as follows:
    
        \noindent\textbf{Electricity Price Transitions:} Since we assume that the dynamics of prices of electricity are exogenous to our AMoD system, $Pr(\boldsymbol{p}(t+1)=\boldsymbol{p_2}|\boldsymbol{p}(t)=\boldsymbol{p}_1,\;\boldsymbol{a}(t)) =$ $Pr(\boldsymbol{p}(t+1)=\boldsymbol{p}_2|\boldsymbol{p}(t)=\boldsymbol{p_1})$, i.e., the dynamics of the price are independent of the action taken. Depending on the setting, new prices might either be deterministic or distributed according to some probability density function at time $t$: $\boldsymbol{p}(t)\sim \mathcal{P}(t)$, which is determined by the electricity provider.
    
    \noindent\textbf{Vehicle Transitions:} For each vehicle at node $i$ and energy level $v$, the transition probability is defined by the action probability vector $\boldsymbol{\alpha_i^v}(t)$. Each vehicle transitions into state $(i,j,1,v-v_{ij})$ with probability $\alpha_{ij}^v(t)$, stays idle in state $(i,i,0,v)$ with probability $\alpha_{ii}^v(t)$ or charges and transitions into state  $(i,i,0,v+1)$ with probability $\alpha_{ic}^v(t)$. The vehicles at intermediate states $(i,j,k,v)$ transition into state $(i,j,k+1,v)$ if $k<\tau_{ij}-1$ or $(j,j,0,v)$ if $k=\tau_{ij}-1$ with probability 1. The total transition probability to the vehicle states $\boldsymbol{s_{veh}}(t+1)$ given $\boldsymbol{s_{veh}}(t)$ and $\boldsymbol{\alpha}(t)$ is the sum of all the probabilities of the feasible transitions from $\boldsymbol{s_{veh}}(t)$ to $\boldsymbol{s_{veh}}(t+1)$ under $\boldsymbol{\alpha}(t)$, where the probability of a feasible transition is the multiplication of individual vehicle transition probabilities (since the vehicle transition probabilities are independent). Note that instead of gradually dissipating the energy of the vehicles on their route, we immediately discharge the required energy for the trip from their batteries and keep them constant during the trip. This ensures that the vehicles have enough battery to complete the ride and does not violate the model, because the vehicles arrive to their destinations with true value of energy and a new action will only be taken when they reach the destination.
    
        \noindent\textbf{Queue Transitions:}
    The queue lengths transition according to the prices and the vehicle routing decisions. For prices $\ell_{ij}(t)$ and induced arrival rate $\Lambda_{ij}(t)$, the probability that $A_{ij}(t)$ new customers arrive in the queue $(i,j)$ is:
    $$
    Pr(A_{ij}(t))=\frac{e^{-\Lambda_{ij}(t)}\Lambda_{ij}(t)^{A_{ij}(t)}}{(A_{ij}(t))!}
    $$
    Let us denote the total number of vehicles routed from node $i$ to $j$ at time $t$ as $x_{ij}(t)$, which is given by: \begin{equation}x_{ij}(t) =\sum_{v=v_{ij}}^{v_{\max}}x_{ij}^v(t)=\sum_{v=v_{ij}}^{v_{\max}}s_{ij1}^{v-{v_{ij}}}(t+1).\end{equation}
    Given $\boldsymbol{s_{veh}}(t+1)$ and  $x_{ij}(t)$, the probability that the queue length $q_{ij}(t+1)=q$ is:
\begin{equation*}
\begin{aligned}
&Pr(q_{ij}(t+1)=q|\boldsymbol{s}(t),\boldsymbol{a}(t),\boldsymbol{s_{veh}}(t+1))=
Pr(A_{ij}(t)=q-q_{ij}(t)+x_{ij}(t)),
\end{aligned}
\end{equation*}
    if $q>0$, and $Pr(A_{ij}(t)\leq-q_{ij}(t)+x_{ij}(t))$ if $q=0$. Since the arrivals are independent, the total probability that the queue vector $\boldsymbol{q}(t+1)=\boldsymbol{q}$ is:
\begin{equation*}
\begin{aligned}
&Pr(\boldsymbol{q}(t+1)=\boldsymbol{q}|\boldsymbol{s}(t),\boldsymbol{a}(t),\boldsymbol{s_{veh}}(t+1))=
\prod_{i\in{\cal M}}\prod_{\substack{j\in{\cal M}\\j\neq i}} Pr(q_{ij}(t+1)|\boldsymbol{s}(t),\boldsymbol{a}(t),\boldsymbol{s_{veh}}(t+1)).
\end{aligned}
\end{equation*}
Hence, the transition probability is defined as:
\begin{equation}
    \label{eq:transition}
    \begin{aligned}
    Pr(\boldsymbol{s}(t+1)|\boldsymbol{s}(t),\boldsymbol{a}(t))=& Pr(\boldsymbol{p}(t+1)|\boldsymbol{p}(t))
    \times Pr(\boldsymbol{s_{veh}}(t+1)|\boldsymbol{s}(t),\boldsymbol{\alpha}(t))\\
    &\times Pr(\boldsymbol{q}(t+1)|\boldsymbol{s}(t),\boldsymbol{\alpha}(t),\boldsymbol{s_{veh}}(t+1))
    \end{aligned}
\end{equation}
We illustrate how the vehicles and queues transition into new states consequent to an action in Figure \ref{fig:transitionschema}.

\begin{figure*}[t]
    \centering
    \includegraphics[width=\textwidth]{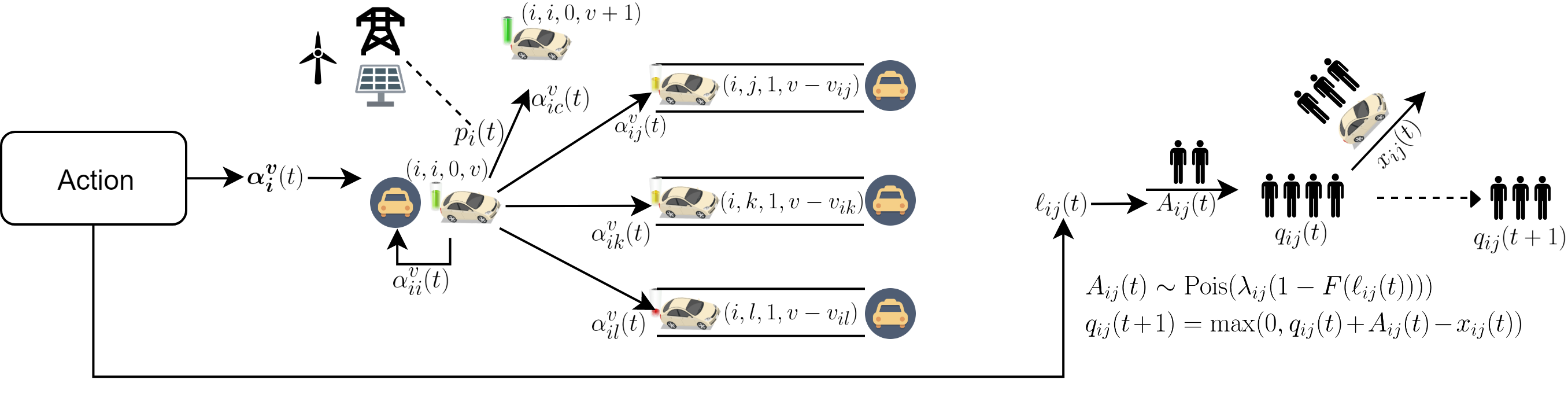}
    \caption{The schematic diagram representing the state transition of our MDP. Upon taking an action, a vehicle at state $(i,i,0,v)$ charges for a price of $p_i(t)$ and transitions into state $(i,i,0,v+1)$ with probability $\alpha_{ic}^v(t)$, stays idle at state $(i,i,0,v)$ with probability $\alpha_{ii}^v(t)$, or starts traveling to another node $j$ and transitions into state $(i,j,1,v-v_{ij})$ with probability $\alpha_{ij}^v(t)$. Furthermore, $A_{ij}(t)$ new customers arrive to the queue $(i,j)$ depending on the price $\ell_{ij}(t)$. After the routing and charging decisions are executed for all the EVs in the fleet, the queues are modified.}
    \label{fig:transitionschema}
\end{figure*}

        \item $r$: The reward function $r(t)$ is a function of state-action pairs at time $t$: $r(t)=r(\boldsymbol{a}(t),\boldsymbol{s}(t))$. Let $x_{ic}^v(t)$ denote the number of vehicles charging at node $i$ starting with energy level $v$ at time period $t$. The reward function $r(t)$ is defined as:
\begin{align}
    \nonumber r(t)=&\sum_{i\in \cal{M}}\sum_{\substack{j\in \cal{M}\\j\neq i}}\ell_{ij}(t)A_{ij}(t)-w\sum_{\in \cal{M}}\sum_{\substack{j\in \cal{M}\\j\neq i}} q_{ij}(t)-\sum_{i\in \cal{M}}\sum_{v=0}^{v_{\max}-1}(\beta+p_i)x_{ic}^v(t)\\&-\beta\sum_{i\in \cal{M}}\sum_{\substack{j\in \cal{M}\\j\neq i}} x_{ij}(t)-\beta\sum_{i\in \cal{M}}\sum_{\substack{j\in \cal{M}\\j\neq i}}\sum_{k=1}^{\tau_{ij}-1}\sum_{v=0}^{v_{\max}-1}s_{ijk}^v(t)
\end{align}
The first term corresponds to the revenue generated by the passengers that request a ride for a price $\ell_{ij}(t)$, the second term is the queue cost of the passengers that have not yet been served, the third term is the charging and operational costs of the charging vehicles and the last two terms are the operational costs of the vehicles making trips. Note that revenue generated is immediately added to the reward function when the passengers enter the network instead of after the passengers are served. Since the reinforcement learning approach is based on maximizing the cumulative reward gained, all the passengers eventually have to be served in order to prevent queues from blowing up and hence it does not violate the model to add the revenues immediately.
        \end{enumerate}

Using the definitions of the tuple $(\mathcal{S}, \mathcal{A}, \mathcal{T}, r)$, we model the dynamic problem as an MDP. Given large-dimensional state and action spaces with infinite cardinality, we can not solve the MDP using exact dynamic programming methods. As a solution, we characterize the real-time policy via a deep neural network and execute reinforcement learning in order to develop a real-time policy.
\subsection{Reinforcement Learning Method}
In this subsection, we go through the preliminaries of reinforcement learning and briefly explain the idea of the algorithm we adopted.
\subsubsection{Preliminaries}
 The real-time policy associated with the MDP is defined as a function parameterized by $\theta$: $$\pi_\theta(\boldsymbol{a}|\boldsymbol{s})=\pi:\mathcal{S}\times \mathcal{A}\rightarrow [0,1],$$ i.e., a probability distribution in the state-action space. Given a state $\boldsymbol{s}$, the policy returns the probability for taking the action $\boldsymbol{a}$ (for all actions), and samples an action according to the probability distribution. The goal is to derive the optimal policy $\pi^*$, which maximizes the discounted cumulative expected rewards $J_{\pi}$:
\begin{equation*}
    \label{eq:discountedrewards}
    \begin{aligned}
        &J_{\pi^*}=\underset{\pi}{\max}\;J_{\pi}=\underset{\pi}{\max}\;\mathbb{E_\pi}\left[\sum_{t=0}^\infty\gamma^tr(t)\right],\\
        &\pi^*=\underset{\pi}{\argmax}\;\mathbb{E_\pi}\left[\sum_{t=0}^\infty\gamma^tr(t)\right],
    \end{aligned}
\end{equation*}
where $\gamma \in (0,1]$ is the discount factor. The value of taking an action $\boldsymbol{a}$ in state $\boldsymbol{s}$, and following the policy $\pi$ afterwards is characterized by the value function $Q_\pi(\boldsymbol{s},\boldsymbol{a})$:
\begin{equation*}
    \label{eq:qfunct}
    Q_\pi(\boldsymbol{s},\boldsymbol{a})=\mathbb{E}_\pi\left[\sum_{t=0}^\infty \gamma^tr(t)|\boldsymbol{s}(0)=\boldsymbol{s},\boldsymbol{a}(0)=\boldsymbol{a}\right].
\end{equation*}
The value of being in state $\boldsymbol{s}$ is formalized by the value function $V_\pi(\boldsymbol{s})$:
\begin{equation*}
    V_\pi(\boldsymbol{s})=\mathbb{E}_{\boldsymbol{a}(0),\pi}\left[\sum_{t=0}^\infty \gamma^t r(t)|\boldsymbol{s}(0)=\boldsymbol{s}\right],
\end{equation*}
and the advantage of taking the action $\boldsymbol{a}$ in state $\boldsymbol{s}$ and following the policy $\pi$ thereafter is defined as the advantage function $A_\pi(\boldsymbol{s},\boldsymbol{a})$:
\begin{equation*}
    A_\pi(\boldsymbol{s},\boldsymbol{a})=Q_\pi(\boldsymbol{s},\boldsymbol{a})- V_\pi(\boldsymbol{s}).
\end{equation*}
The methods used by reinforcement learning algorithms can be divided into three main groups: 1) critic-only methods, 2) actor-only methods, and 3) actor-critic methods, where the word critic refers to the value function and the word actor refers to the policy \cite{surveyrl}. Critic-only (or value-function based) methods (such as Q-learning \cite{qlearning1} and SARSA \cite{sarsa}) improve a deterministic policy using the value function by iterating:
\begin{equation*}
\begin{aligned}
    &\boldsymbol{a^*}=\underset{\boldsymbol{a}}{\argmax}\;Q_{\pi}(\boldsymbol{s},\boldsymbol{a}),\\
    &\pi(\boldsymbol{a^*}|\boldsymbol{s})\longleftarrow 1.
\end{aligned}
\end{equation*}
Actor-only methods (or policy gradient methods), such as Williams' REINFORCE algorithm \cite{reinforce}, improve the policy by updating the parameter $\theta$ by gradient ascent, without using any form of a stored value function:
\begin{equation*}
    \theta(t+1)=\theta(t)+\alpha\nabla_\theta \mathbb{E}_{\pi_{\theta(t)}}\left[\sum_\tau\gamma^\tau r(\tau)\right].
\end{equation*}
The advantage of policy gradient methods is their ability to generate actions from a continuous action space by utilizing a parameterized policy.

Finally, actor-critic methods \cite{actorcritic1,actorcritic2} make use of both the value functions and policy gradients:
\begin{equation*}
    \theta(t+1)=\theta(t)+\alpha\nabla_\theta \mathbb{E}_{\pi_{\theta(t)}}\left[Q_{\pi_{\theta(t)}}(\boldsymbol{s},\boldsymbol{a})\right].
\end{equation*}
Actor-critic methods are able to produce actions in a continuous action space, while reducing the high variance of the policy gradients by adding a critic (value function).

All of these methods aim to update the parameters $\theta$ (or directly update the policy $\pi$ for critic-only methods) to improve the policy. In deep reinforcement learning, the policy $\pi$ is defined by a deep neural network, whose weights constitute the parameter $\theta$. To develop a real-time policy for our MDP, we adopt a practical policy gradient method called Proximal Policy Optimization (PPO).
\subsubsection{Proximal Policy Optimization}
PPO is a practical policy gradient method developed in \cite{schulman2017ppo}, and is effective for optimizing large non-linear policies such as deep neural networks. It preserves some of the benefits of trust region policy optimization (TRPO) \cite{schulman2015trpo} such as monotonic improvement, but is much simpler to implement because it can be optimized by a first-order optimizer, and is empirically shown to have better sample complexity.

In TRPO, an objective function (the ``surrogate'' objective) is maximized subject to a constraint on the size of the policy update so that the new policy is not too far from the old policy:
\begin{subequations}
\begin{align}
&\underset{\theta}{\text{maximize}}& & \hat{\mathbb{E}}_t\left[\frac{\pi_\theta(\bm a_t|\bm s_t)}{\pi_{\theta_{old}}(\bm a_t|\bm s_t)}\hat{A}_t\right] \\
& \text{subject to}
& & \hat{\mathbb{E}}_t\left[\textnormal{KL}\left[\pi_{\theta_{old}}(\cdot|\bm s_t),\pi_\theta(\cdot|\bm s_t)\right]\right]\leq \delta,
\end{align}
\end{subequations}
where $\pi_\theta$ is a stochastic policy and $\hat{A}_t$ is an estimator of the advantage function at timestep $t$. The expectation $\hat{\mathbb{E}}_t[\dots]$ indicates the empirical average over a finite batch of samples and $\textnormal{KL}\left[\pi_{\theta_{old}}(\cdot|\bm s_t),\pi_\theta(\cdot|\bm s_t)\right]$ denotes the Kullback–Leibler divergence between $\pi_{\theta_{old}}$ and $\pi$. Although TRPO  solves the above constrained maximization problem using conjugate gradient, the theory justifying TRPO actually suggests using a penalty instead of a constraint, i.e., solving the unconstrained optimization problem
\begin{equation}
    \underset{\theta}{\textnormal{maximize}}~\hat{\mathbb{E}}_t\left[\frac{\pi_\theta(\bm a_t|\bm s_t)}{\pi_{\theta_{old}}(\bm a_t|\bm s_t)}\hat{A}_t-\beta\textnormal{KL}\left[\pi_{\theta_{old}}(\cdot|\bm s_t),\pi_\theta(\cdot|\bm s_t)\right] \right],
\end{equation}
for some penalty coefficient $\beta$. TRPO uses a hard constraint rather than a penalty because it is hard to choose a single value of $\beta$ that performs well. To overcome this issue and develop a  first-order algorithm that emulates the monotonic improvement of TRPO (without solving the constrained optimization problem), two PPO algorithms are constructed by: 1) clipping the surrogate objective and 2) using adaptive KL penalty coefficient \cite{schulman2017ppo}.
\begin{enumerate}
    \item \textit{Clipped Surrogate Objective:} Let $r_t(\theta)$ denote the probability ratio $r_t(\theta)=\frac{\pi_\theta(\bm a_t|\bm s_t)}{\pi_{\theta_{old}}(\bm a_t|\bm s_t)}$, so $r(\theta_{old})=1$. TRPO maximizes
    \begin{equation}
        L(\theta)=\hat{\mathbb{E}}_t\left[\frac{\pi_\theta(\bm a_t|\bm s_t)}{\pi_{\theta_{old}}(\bm a_t|\bm s_t)}\hat{A}_t\right]=\hat{\mathbb{E}}_t\left[r_t(\theta)\hat{A}_t\right].
    \end{equation}
    subject to the KL divergence constraint. Without a constraint however this would lead to a large policy update. To prevent this, PPO modifies the surrogate objective to penalize changes to the policy that move $r_t(\theta)$ away from 1:
    \begin{equation}
        L^{CLIP}(\theta)=\hat{\mathbb{E}}_t\left[\min(r_t(\theta)\hat{A}_t,\textnormal{clip}(r_t(\theta),1-\epsilon,1+\epsilon)\hat{A}_t)\right],
    \end{equation}
    where $\epsilon$ is a hyperparameter, usually 0.1 or 0.2. The term $\textnormal{clip}(r_t(\theta),1-\epsilon,1+\epsilon)\hat{A}_t)$ modifies the surrogate objective by clipping the probability ratio, which removes the incentive for moving $r_t$ outside of the interval $[1-\epsilon,1+\epsilon]$. By taking the minimum of the clipped and the unclipped objective, the final objective becomes a lower bound on the unclipped objective.
    \item \textit{Adaptive KL Penalty Coefficient:} Another approach is to use a penalty on KL divergence and to adapt the penalty coefficient so that some target value of the KL divergence $d_{\textnormal{targ}}$ is achieved at each policy update. In each policy update, the following steps are performed:
    \begin{itemize}
        \item Using several epochs of minibatch SGD, optimize the KL-penalized objective
\begin{equation}
    L^{KLPEN}(\theta)=\hat{\mathbb{E}}_t\left[\frac{\pi_\theta(\bm a_t|\bm s_t)}{\pi_{\theta_{old}}(\bm a_t|\bm s_t)}\hat{A}_t-\beta\textnormal{KL}\left[\pi_{\theta_{old}}(\cdot|\bm s_t),\pi_\theta(\cdot|\bm s_t)\right] \right]
\end{equation}
\item Compute $d=\hat{\mathbb{E}}_t\left[\textnormal{KL}\left[\pi_{\theta_{old}}(\cdot|\bm s_t),\pi_\theta(\cdot|\bm s_t)\right]\right]$
\begin{itemize}
    \item If $d<d_{\textnormal{targ}}/1.5,~\beta\leftarrow \beta/2$
    \item If $d>d_{\textnormal{targ}}\times 1.5,~\beta\leftarrow \beta\times 2$.
\end{itemize}
    \end{itemize}
    The updated $\beta$ is then used for the next policy update. This scheme allows $\beta$ to adjust if KL divergence is significantly different than $d_{\textnormal{targ}}$ so that the desired KL divergence between the old and the updated policy is attained.
\end{enumerate}
A PPO algorithm using fixed-length trajectory segments is summarized in Algorithm~\ref{alg:ppo}. Each iteration, each of $N$ (parallel) actors collect $T$ timesteps of data. Then the surrogate loss on these $NT$ timesteps of data is constructed and optimized  with minibatch SGD for $K$ epochs.
\begin{algorithm}[h]
\SetAlgoLined
 
 \For{$\textnormal{iteration}=0,1,2,\dots$}{
 \For{$\textnormal{actor}=1,2,\dots, N$}{
   Run policy $\pi_{\theta_{old}}$ in environment for $T$ timesteps.\\
   Compute advantage estimates $\hat{A}_1,\dots,\hat{A}_T$}
    Optimize surrogate $L^{CLIP}$ or $L^{KLPEN}$ w.r.t. $\theta$, with $K$ epochs and minibatch size $M\leq NT$.\\
    $\theta_{old}\leftarrow\theta$
    }
 \caption{PPO, Actor-Critic Style}\label{alg:ppo}
\end{algorithm}

In this work, we used the PPO algorithm with the clipped surrogate objective, because experimentally it it shown to have better performance than the PPO algorithm with adaptive KL penalty coefficient \cite{schulman2017ppo}. We refer the reader to \cite{schulman2017ppo} for a comprehensive study on PPO algorithms. 

In the next section, we present our numerical studies demonstrating the performance of the RL policy.
}

\section{Numerical Study}\label{sec:numerical}
In this section, we discuss the numerical experiments and results for the performance of reinforcement learning approach to the dynamic problem and compare with the performance of several static policies, including the optimal static policy outlined in Section \ref{sec:static}. We solved for the optimal static policy using CVX, a package for specifying and solving convex programs \cite{cvx}. To implement the dynamic environment compatible with reinforcement learning algorithms, we used Gym toolkit \cite{1606.01540} developed by OpenAI to create an environment. For the  implementation of the \bt{PPO} algorithm, we used Stable Baselines toolkit \cite{stable-baselines}.

We chose an operational cost of $\beta=\$0.1$ (by normalizing the average price of an electric car  over 5 years \cite{avgevprice}) and maximum willingness to pay  $\ell_{\max}=\$30$. For prices of electricity $p_i(t)$, we generated random prices for different locations and different times using the statistics of locational marginal prices in \cite{electricityprices}. We chose a maximum battery capacity of $20$kWh. We discretrized the battery energy into 5 units, where one unit of battery energy is $4$kWh. The time it takes to deliver one unit of charge is taken as one time epoch, which is equal to $5$ minutes in our setup. The waiting time cost for one period is $w=\$2$ (average hourly wage is around $\$24$ in the United States \cite{avgwages}). 

Note that the dimension of the state space grows significantly with battery capacity $v_{\max}$, because it expands the states each vehicle can have by $v_{\max}$.
Therefore, for computational purposes, we conducted two case studies: 1) Non-electric AMoD case study with a larger network in Manhattan, 2) Electric AMoD case study with a smaller network in San Francisco. \btt{We picked two different real world networks in order to demonstrate the universality of reinforcement learning method  in establishing a real-time policy. In particular, our intention is to support the claim that the success of the reinforcement learning method is not restricted to a single network, but generalizes to multiple real world networks.} Both experiments were performed on a laptop computer with Intel$^\textnormal{\textregistered}$ Core$^\textnormal{TM}$ i7-8750H CPU (6$\times$2.20 GHz) and 16 GB DDR4 2666MHz RAM.
\subsection{Case Study in Manhattan}\label{sec:man}

In a non-electric AMoD network, the energy dimension $v$ vanishes. Because there is no charging action\footnote{The vehicles still refuel, however this takes negligible time compared to the trip durations.}, we can perform coarser discretizations of time. Specifically, we can allow each discrete time epoch to cover $5\times\underset{\substack{i,j|i\neq j}}{\min\;}\tau_{ij}$ minutes, and normalize the travel times $\tau_{ij}$ and $w$ accordingly
(For EV's, because charging takes a non-negligible but shorter time than travelling, in general we have $\tau_{ij}>1$, and larger number of states). The static profit maximization problem in \eqref{eq:flowoptimization} for AMoD with non-electric vehicles can  be rewritten as:

\begin{equation}
\label{eq:nonelectricflowoptimization}
\begin{aligned}
&\underset{x_{ij},\ell_{ij}}{\text{max}}
& &\sum_{i\in{\cal M}}\sum_{j\in{\cal M}}\lambda_{ij}\ell_{ij}(1-F(\ell_{ij}))-\beta_g \sum_{i\in{\cal M}}\sum_{j\in{\cal M}} x_{ij} \tau_{ij} \\
& \text{subject to}
& &\lambda_{ij}(1-F(\ell_{ij})) \leq x_{ij} \quad\forall i,j \in \mathcal M,\\
& & &\sum_{j\in{\cal M}} x_{ij}=\sum_{j\in{\cal M}} x_{ji}\quad\forall i\in\mathcal M,\\
& & & x_{ij}\geq 0\; ~\forall i,j\in \mathcal M.
\end{aligned}
\end{equation}


The operational costs $\beta_g=\$2.5$ (per 10 minutes, \cite{avgcostofdriving}) are different than those of electric vehicles. Because there is no ``charging" (or refueling action, since it takes negligible time), $\beta_g$ also includes fuel cost. The optimal static policy is used to compare and highlight the performance of the real-time policy\footnote{\label{ft:convertstaticactions}The solution of the static problem yields vehicle flows. In order to make the policy compatible with our environment and to generate integer actions that can be applied in a dynamic setting, we randomized the actions by dividing each flow for OD pair $(i,j)$ (and energy level $v$) by the total number of vehicles in $i$ (and energy level $v$) and used that fraction as the probability of sending a vehicle from $i$ to $j$ (with energy level $v$).}.

\begin{wrapfigure}{r}{0.16\textwidth}
\centering
    \includegraphics[width=.18\textwidth]{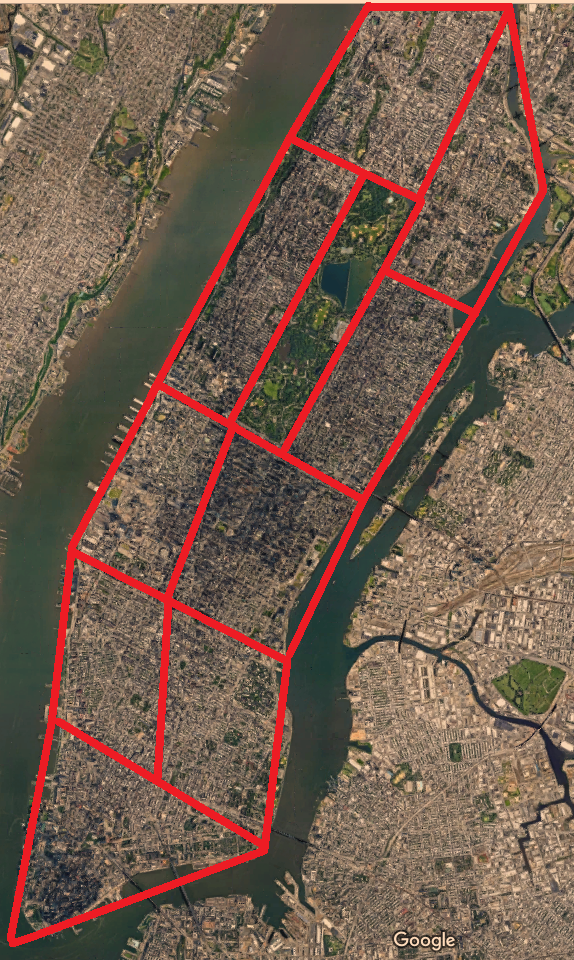}
    \vspace{-.75cm}
  \caption{Manhattan divided into $m=10$ regions.}
  \label{fig:manhattanregions}
  \vspace{-.75cm}
\end{wrapfigure}
We divided Manhattan into 10 regions as in Figure \ref{fig:manhattanregions}, and using the yellow taxi data from the New York City Taxi and Limousine Commission
dataset \cite{manhattantaxidata} for May 04, 2019, Saturday between 18.00-20.00, we extracted the average arrival rates for rides and trip durations  between the regions (we exclude the rides occurring in the same region). We trained our model by creating new induced random arrivals with the same average arrival rate using prices determined by our policy. For the fleet size, we used a fleet of 1200 autonomous vehicles (according to the optimal fleet size emerging from the static problem).

For training, we used a neural network with \bt{4} hidden layers and \bt{128} neurons in each hidden layer. The rest of the parameters are left as default as specified by the Stable Baselines toolkit \cite{stable-baselines}. \bt{In order to get the best policy, we train 3 different models using DDPG\cite{lillicrap2015ddpg}, TRPO\cite{schulman2015trpo}, and PPO.  We trained the models for 10 million iterations, and the performances of the trained models are summarized in Table~\ref{tab:rlalgs} using average rewards and queue lengths as metrics. Our experiments indicate that the model trained using PPO is performing the best among the three, hence we use that model as our real-time policy.}
\definecolor{lightblue}{RGB}{135,206,250}

\begin{table}[h]
    \centering
\bt{
    \begin{tabular}{|c|c|c|c|}
    \hline
        \backslashbox{Metrics}{Algorithms} & DDPG & TRPO & PPO  \\\hline
        Average Rewards& 9825.69& 13142.47&\textbf{15527.34}\\\hline
        Average Queue Length&  431.76&87.96 &\textbf{68.11} \\\hline
    \end{tabular}}
    \caption{Performances of RL policies trained with different algorithms.}
    \label{tab:rlalgs}
\end{table}


We compare different policies' performance using the rewards and total queue length as metrics. The results are demonstrated in Figure \ref{fig:manhattanresults}.
In Figure \ref{fig:staticvsdynamic} we compare the rewards generated and the total queue length by applying the static and the real-time policies as defined in Sections \ref{sec:static} and \ref{sec:dynamic}.  We can observe that while the optimal static policy provides rate stability in a dynamic setting (since the queues do not blow up),  it fails to generate profits as it is not able to clear the queues. On the other hand, the real-time policy is able to keep the total length of the queues 100 times shorter than the static policy while generating higher profits.

\begin{figure*}  
\begin{subfigure}[b]{0.4\textwidth}
\centering
\includegraphics[width=\linewidth]{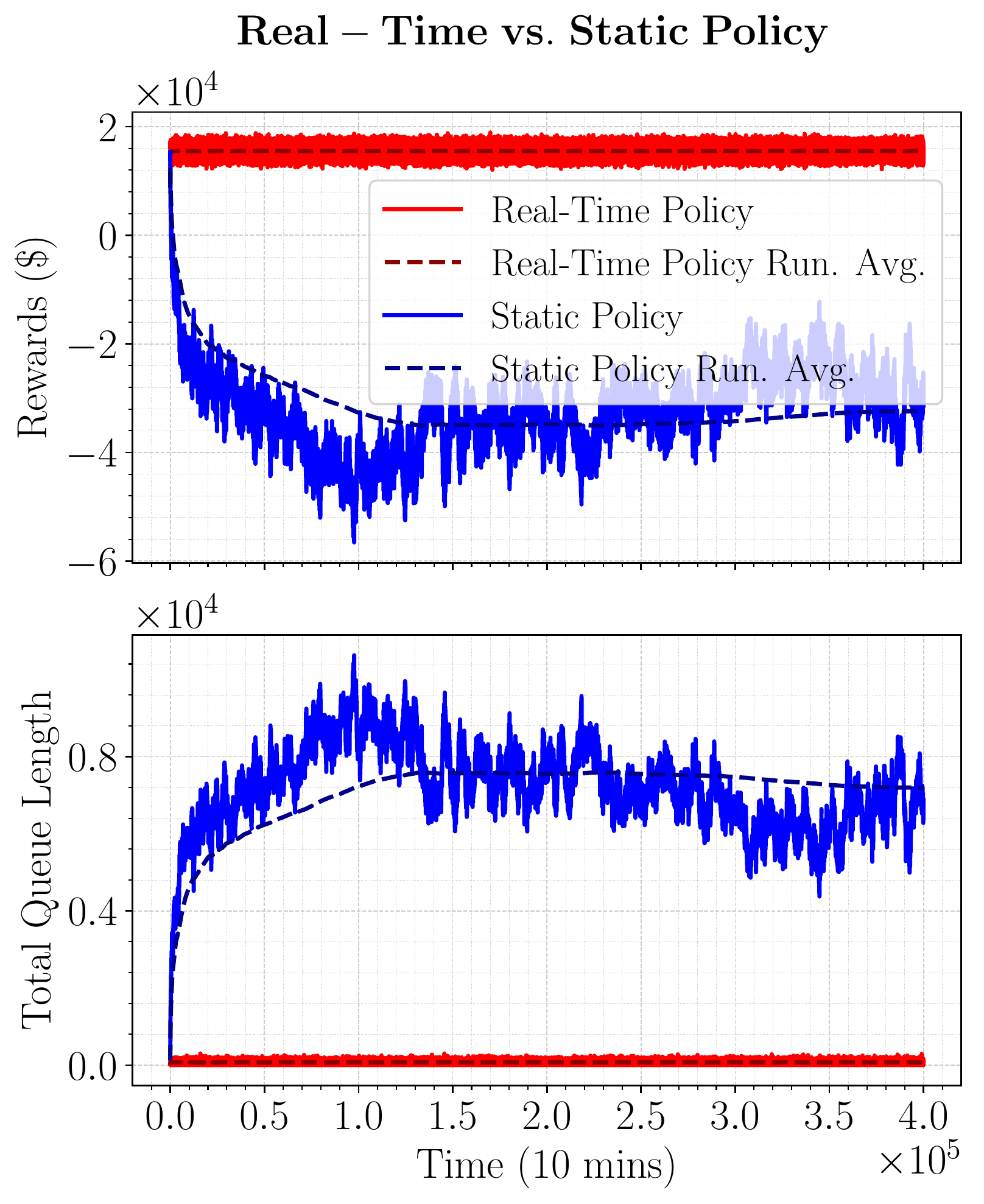}
\caption{}
\label{fig:staticvsdynamic}
\end{subfigure}
\hfill 
\begin{subfigure}[b]{0.4\textwidth}
\centering
\includegraphics[width=\linewidth]{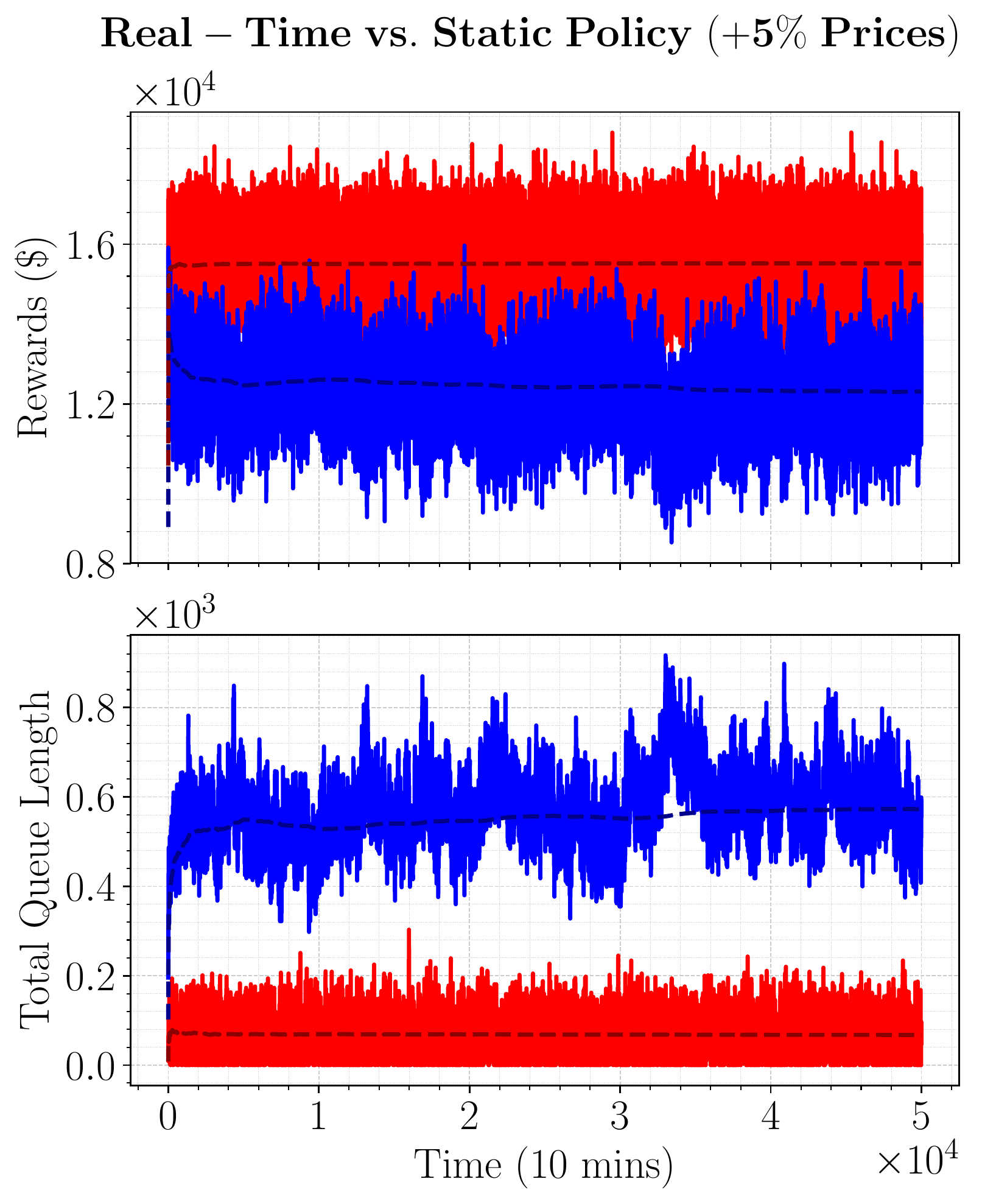}
\caption{}
\label{fig:staticvsdynamichigher}
\end{subfigure}
\bigskip  
\begin{subfigure}[b]{0.4\textwidth}
\centering
\includegraphics[width=\linewidth]{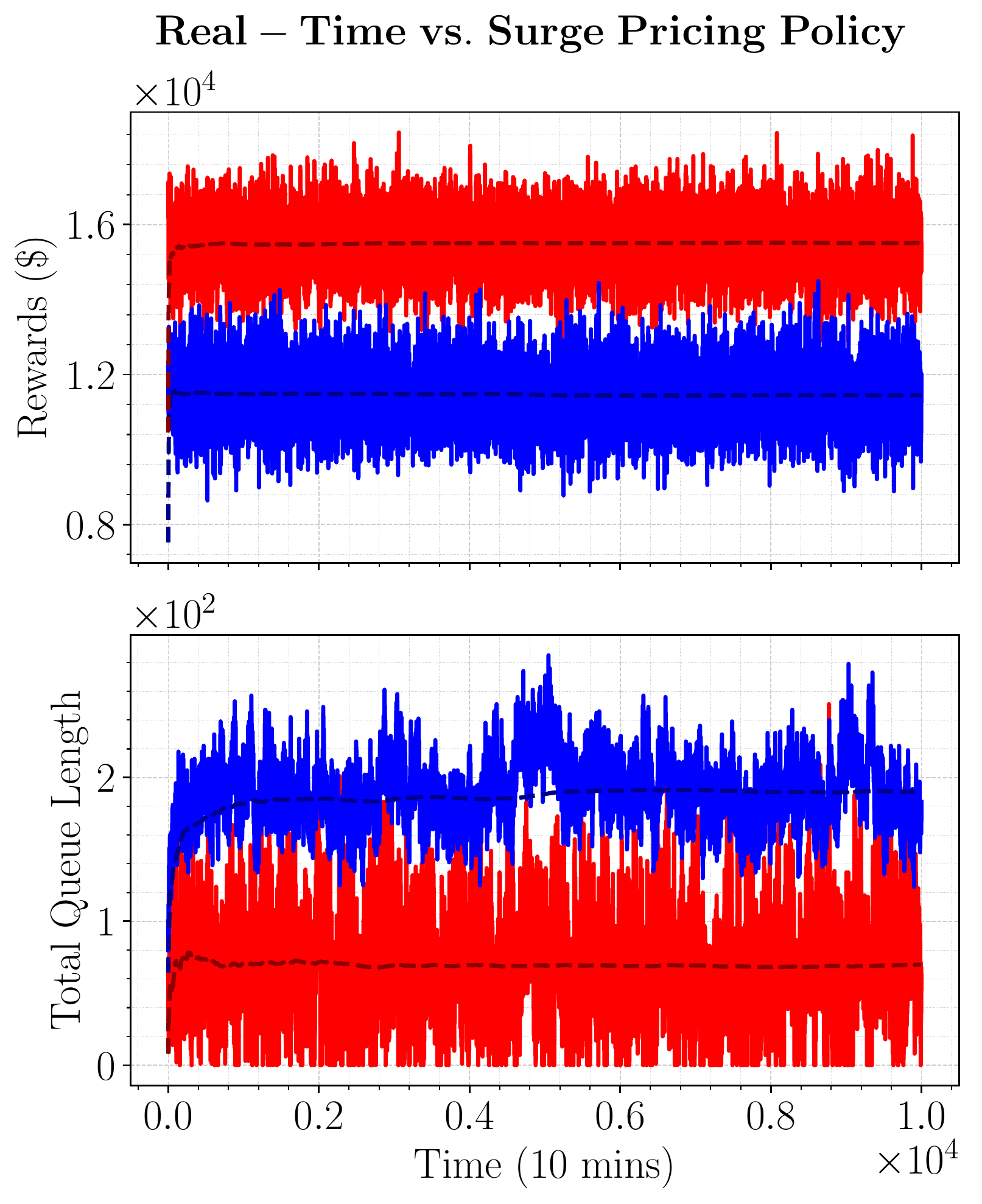}
\caption{}
\label{fig:staticvsdynamicsurge}
\end{subfigure}
\hfill 
\begin{subfigure}[b]{0.4\textwidth}
\centering
\includegraphics[width=\linewidth]{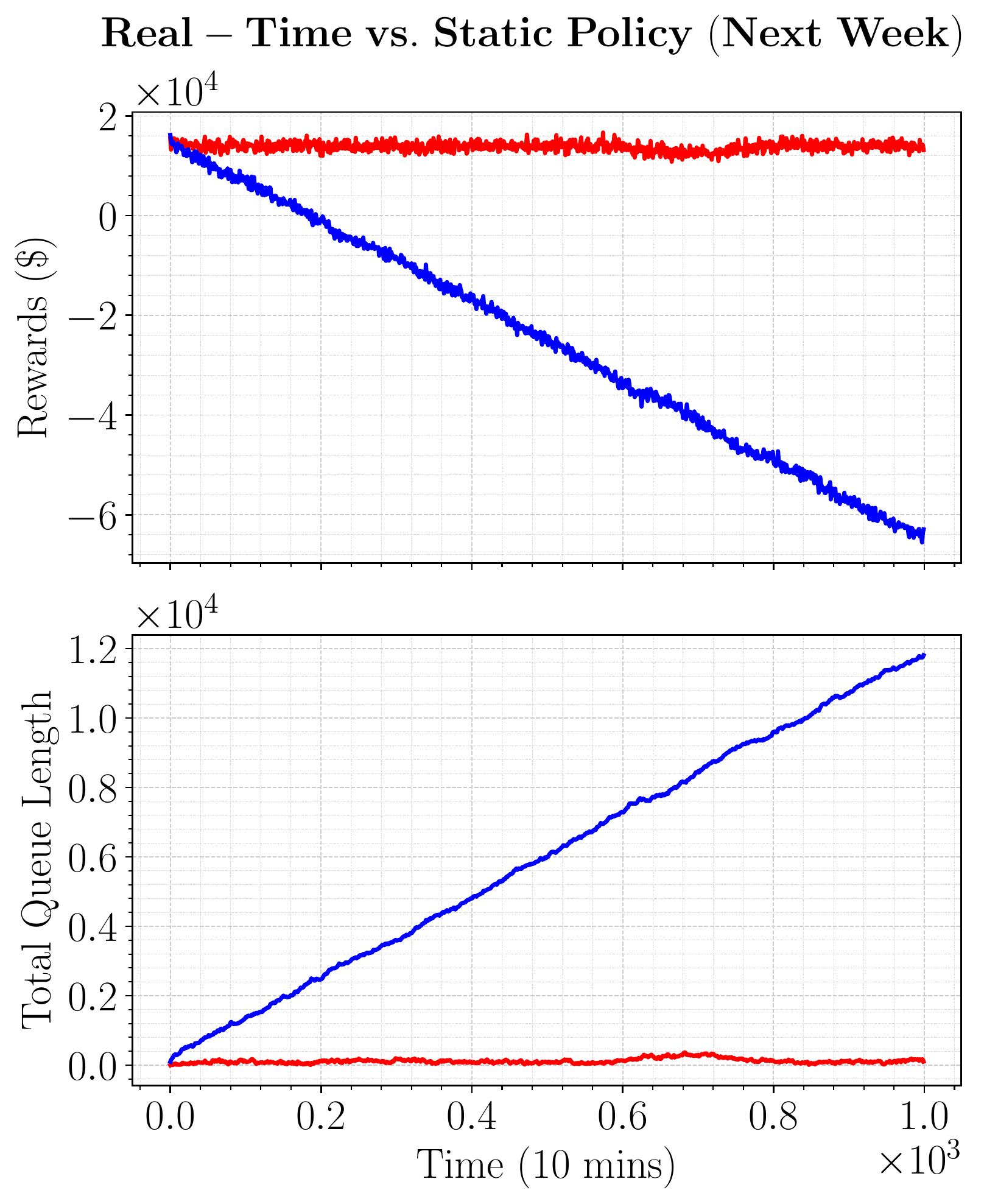}
\caption{}
\label{fig:staticvsdynamicnextweek}
\end{subfigure}
\caption{Comparison of different policies for Manhattan case study. The legends for all figures are the same as the top left figure, where red lines correspond to the real-time policy and blue lines correspond to the static policies (We excluded the running averages for (d), because the static policy diverges). In all scenarios, we use the rewards generated and the total queue length as metrics. In (a), we demonstrate the results from applying the real-time policy and the optimal static policy. In (b), we compare the real-time policy with the static policy that utilizes $5\%$ higher prices than the optimal static policy. In (c), we utilize a surge pricing policy along with the optimal static policy and compare with the real-time policy. In (d), we employ the real-time policy and static policy developed for May 4, 2019, Saturday for the arrivals on May 11, 2019, Saturday.}
\label{fig:manhattanresults}
\end{figure*}

The optimal static policy fails to generate profits and is not necessarily the best static policy to apply in a dynamic setting. As such, in Figure \ref{fig:staticvsdynamichigher} we demonstrate the performance of a sub-optimal static policy, \bt{where the prices are $5\%$ higher than the optimal static prices} to reduce the arrival rates and hence reduce the queue lengths.  Observe that the profits generated are higher than the profits generated using optimal static policy for the static planning problem while the total queue length is less. This result indicates that under the stochasticity of the dynamic setting, a sub-optimal static policy can perform better than the optimal static policy. \bt{ Furthermore, we summarize the performances of other static policies with higher static prices, namely with $5\%,10\%,20\%30\%$, and $40\%$ higher prices than the optimal static prices in Table~\ref{tab:manhattan_higherpricepolicies}. Among these, an increase of $10\%$ performs the best in terms of rewards. Nevertheless, this policy does still do worse in terms of rewards and total queue length compared to the real-time policy, which generates around $10\%$ more rewards and results in $70\%$ less queues. Lastly we note that although a $40\%$ increase in prices results in minimum average queue length, this is a result of significantly reduced induced demand and therefore it generates very low rewards.}
\begin{table}[h]
    \centering
    \bt{
    \begin{tabular}{|c||c|c|c|c|c|}
    \hline
         \backslashbox{Metrics}{\% of opt. static prices}&105\% &110\% &120\% &130\%&140\%  \\\hline
         Average Rewards&12234.13 &\textbf{14112.77} & 13739.35&12046.91 &9625.82\\\hline
         Average Queue Length& 584.05& 231.93&74.64 & 30.88&\textbf{14.20}\\\hline
    \end{tabular}}
    \caption{Performances of static pricing policies for Manhattan case study.}
    \label{tab:manhattan_higherpricepolicies}
\end{table}

Next, we showcase \bt{that} even some heuristic modifications \bt{which} resemble what is done in practice can do better than the optimal static policy. We utilize the optimal static policy, but additionally utilize a surge-pricing policy. The surge-pricing policy aims to decrease the arrival rates for longer queues so that the queues will stay shorter and the rewards will increase. \bt{At each time period, for all OD pairs, the policy is to increase the price by $50\%$ if the queue is longer than $100\%$ of the induced arrival rate. The results are displayed in Figure \ref{fig:staticvsdynamicsurge}. New arrivals bring higher revenue per person and the total queue length is decreased, which stabilizes the network while generating more profits than the optimal static policy. The surge pricing policy results in stable short queues and higher rewards compared to the optimal static policy for the static setting, however, both the real-time policy and the static pricing policy with $10\%$ higher prices are superior. Performances of other surge pricing policies that multiply the prices by $1.25/1.5/2$ if the queue is longer than $50\%/100\%/200\%$ of the induced arrival rates can be found in Table~\ref{tab:manhattan_surgepolicies}. Accordingly, the best surge pricing policy maximizing the rewards is to multiply the prices by $1.25$ if the queue is longer than $50\%$ of the induced arrival rate. Yet, our real-time policy still generates around $20\%$ more rewards and results in $32\%$ less queues. We note that a surge pricing policy that multiplies the prices by $2$ when the queues are longer than $50\%$ of the induced arrival rates minimizes the queues by decreasing the induced arrival rates significantly, which results in substantially low rewards.}

\begin{table}[h]
\centering
\bt{
 \begin{tabular}{|c||c|c|c|c|c|c|} 
 \hline
\multirow{2}{*}{\backslashbox{Surge Multiplier}{Queue Threshold}} &\multicolumn{2}{c|}{50\%}  
& \multicolumn{2}{c|}{100\%} 
& \multicolumn{2}{c|}{200\%}\\ \cline{2-7}
&Queue&Rewards &Queue&Rewards&Queue&Rewards\\ \hline
1.25&101.25 &\textbf{13022.83} &186.56 &12897.30 &380.34  &12357.33\\\hline
1.5& 91.89&12602.90 &178.22 &12589.71 &370.18 &12233.95\\\hline
 2&\textbf{83.15} &5272.04 &162.99 &6224.69 &337.01 &7485.75\\\hline
\end{tabular}}
\caption{Performances of surge pricing policies for Manhattan case study.}
\label{tab:manhattan_surgepolicies}
\end{table}


Finally, we test how the static and the real-time policies are robust to variations in input statistics. We compare the rewards generated and the total queue length applying the static and the real-time policies for the arrival rates of May 11, 2019, Saturday between 18.00-20.00. The results are displayed in Figure \ref{fig:staticvsdynamicnextweek}. Even though the arrival rates between May 11 and May 4 do not differ much, the static policy is not resilient and fails to stabilize when there is a slight change in the network. The real-time policy, on the other hand, is still able to stabilize the network and generate profits. The neural-network based policy is able to determine the correct pricing and routing decisions by considering the current state of the network, even under different arrival rates.

These experiments show us that we can indeed develop a real-time policy using deep reinforcement learning and this policy is resilient to small changes in the network parameters. The next study investigates the idea of generality, i.e., whether we can develop a global real-time policy and fine-tune it to a specific environment with \emph{few-shots} of training, rather than developing a new policy from scratch.

\begin{figure}[t]
\begin{minipage}[t]{.65\textwidth}
\centering
\includegraphics[width=\textwidth]{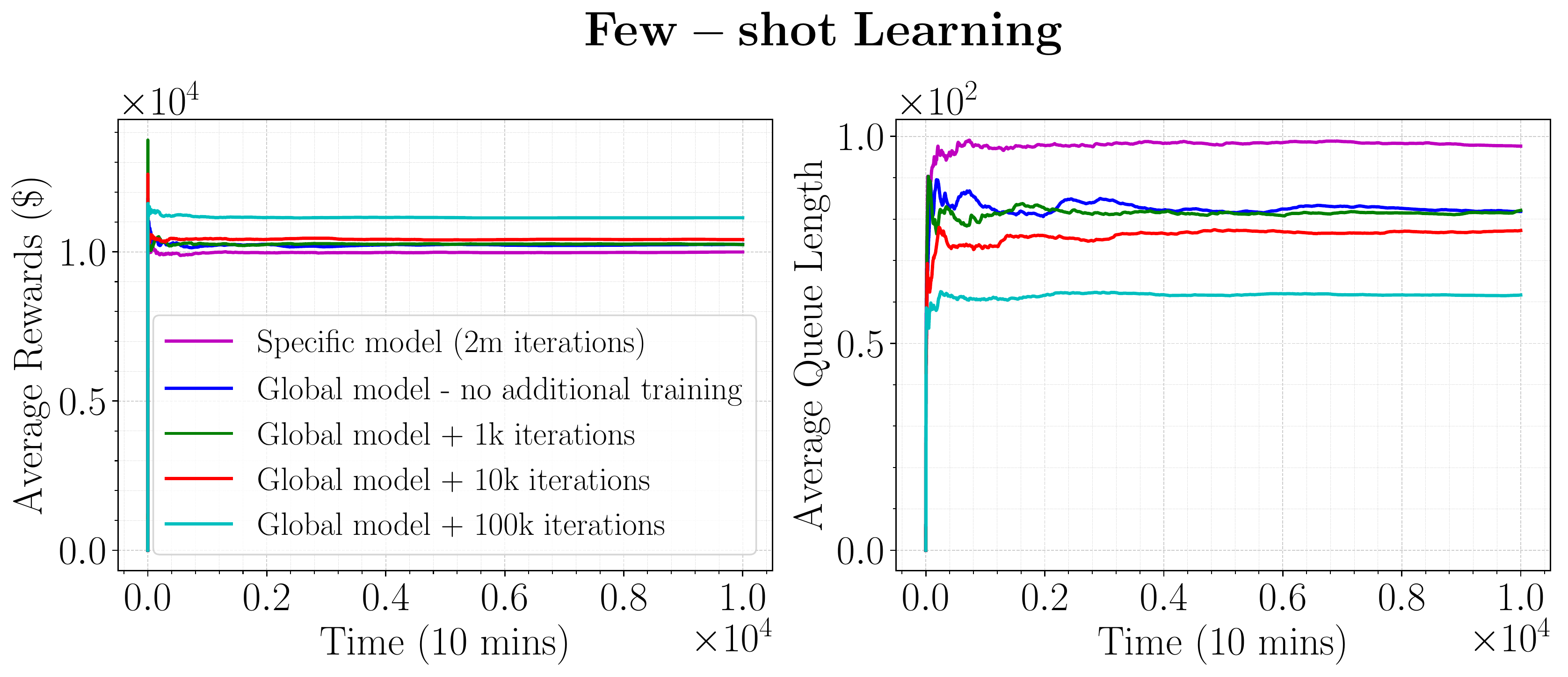}
\caption{Performances of the specific model that is trained from scratch and fine-tuned global model (for different amounts of fine-tuning as specified in the legend): rewards (left) and queue lengths (right).}\label{fig:metaresults}
\end{minipage}\hspace{.2cm}%
\begin{minipage}[t]{.3\textwidth}
\centering
\includegraphics[width=\textwidth]{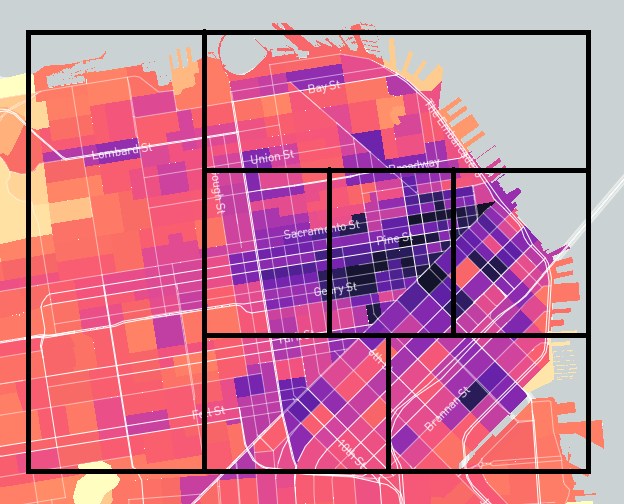}
  \caption{San Francisco divided into $m=7$ regions. Map obtained from the San Francisco County Transportation Authority \cite{sfmap}.}
  \label{fig:sfregions}
\end{minipage}
\vspace{-1.5em}
\end{figure}

\vspace{.5em}
\noindent
\textbf{Few-shot Learning:}
A common problem with reinforcement learning approaches is that because the agent is trained for a specific environment, it fails to respond to a slightly changed environment. Hence, one would need to train a different model for different environments (different network configurations, different arrival rates). However, this is not a feasible solution considering that training one model takes millions of iterations. As a more tractable solution, one could train a global model using different environments, and then calibrate it to the desired environment with fewer iterations rather than training a new model from scratch. We tested this phenomenon by training a global model for Manhattan using various arrival rates and network configurations that we extracted from different 2-hour intervals (We trained the global model for 10 million iterations). We then trained this model for the network configuration and arrival rates on May 6, 2019, Monday between 15.00-17.00. The results are displayed in Figure~\ref{fig:metaresults}. Even with no additional training, the
global model performs better than the specific model trained from scratch for 2 million iterations. Furthermore, with only few iterations, it is possible to improve the performance of the global model significantly. This is an anticipated result, because although the network configurations and arrival rates for different 2-hour intervals are different, the environments are not fundamentally different (the state transitions are governed by similar random processes) and hence it is possible to generalize a global policy and fine-tune it to the desired environment with fewer number of iterations.


\subsection{Case Study in San Francisco}
\begin{wrapfigure}{R}{0.45\textwidth}
\centering
    \includegraphics[width=.45\textwidth]{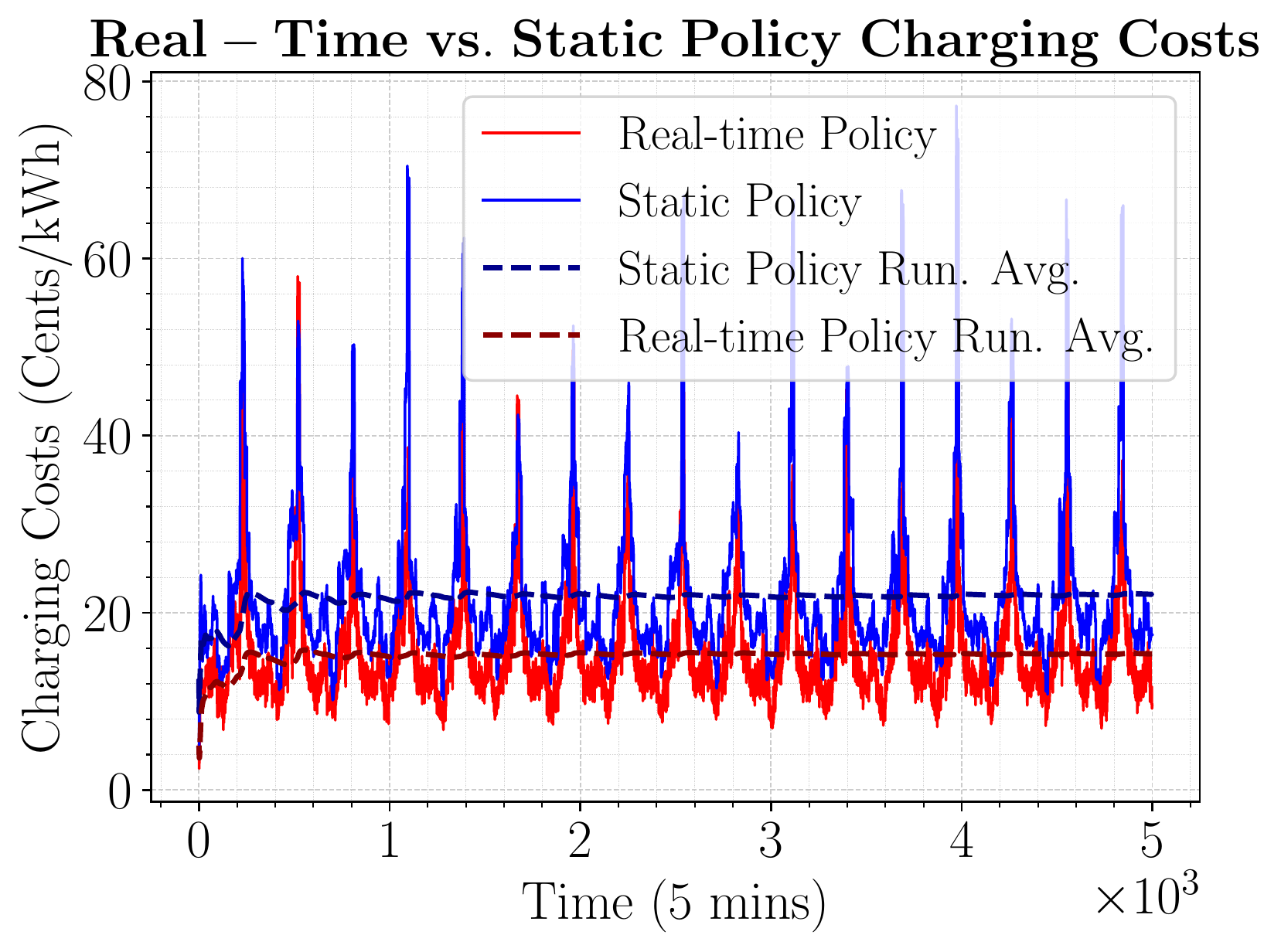}
  \caption{Charging costs for the optimal static policy and the real-time policy in San Francisco case study.}
  \label{fig:sf_chargecostl}
\end{wrapfigure}
We conducted the case study in San Francisco by utilizing an EV fleet of 420 vehicles. We divided San Francisco into 7 regions as in Figure \ref{fig:sfregions}, and using the traceset of mobility of taxi cabs data from CRAWDAD \cite{epfl-mobility-20090224}, we obtained the average arrival rates and travel times between regions (we exclude the rides occurring in the same region).

In Figure \ref{fig:sf_chargecostl}, we compare the charging costs paid under the real-time policy and the static policy. The static policy is generated by using the average value of the electricity prices, whereas the real-time policy takes into account the current electricity prices before executing an action. Therefore, the real-time policy provides cheaper charging options by utilizing smart charging strategies, decreasing the average charging costs by $25\%$.

In Figure \ref{fig:staticvsdynamic_sf}, we compare the rewards and the total queue length resulting from the real-time policy and the static policy. \bt{In Figure \ref{fig:staticvsdynamichigher_sf}, we compare the RL policy to the static policy with $5\%$ higher prices than the optimal static policy, and summarize performances of several other static pricing policies in Table~\ref{tab:bayarea_higherpricepolicies}.
In Figure \ref{fig:staticvsdynamicsurge_sf}, we use the static policy but also utilize a surge pricing policy that multiplies the prices by $1.5$ if the queues are longer than $100\%$ of the induced arrival rates. The performances of other surge pricing policies are also displayed in Table~\ref{tab:bayarea_surgepolicies}}. Similar to the case study in Manhattan, the results demonstrate that the performance of the trained real-time policy is superior to the other policies. \bt{In particular, the RL policy is able to generate around $24\%$ more rewards and result in around $75\%$ less queues than the best heuristic policy, which utilizes $30\%$ higher static prices than the optimal static policy}.

\begin{figure}[t]
     \centering
     \hspace{-0.5cm}
     \begin{subfigure}[y]{0.33\textwidth}
         \centering
         \includegraphics[width=\textwidth]{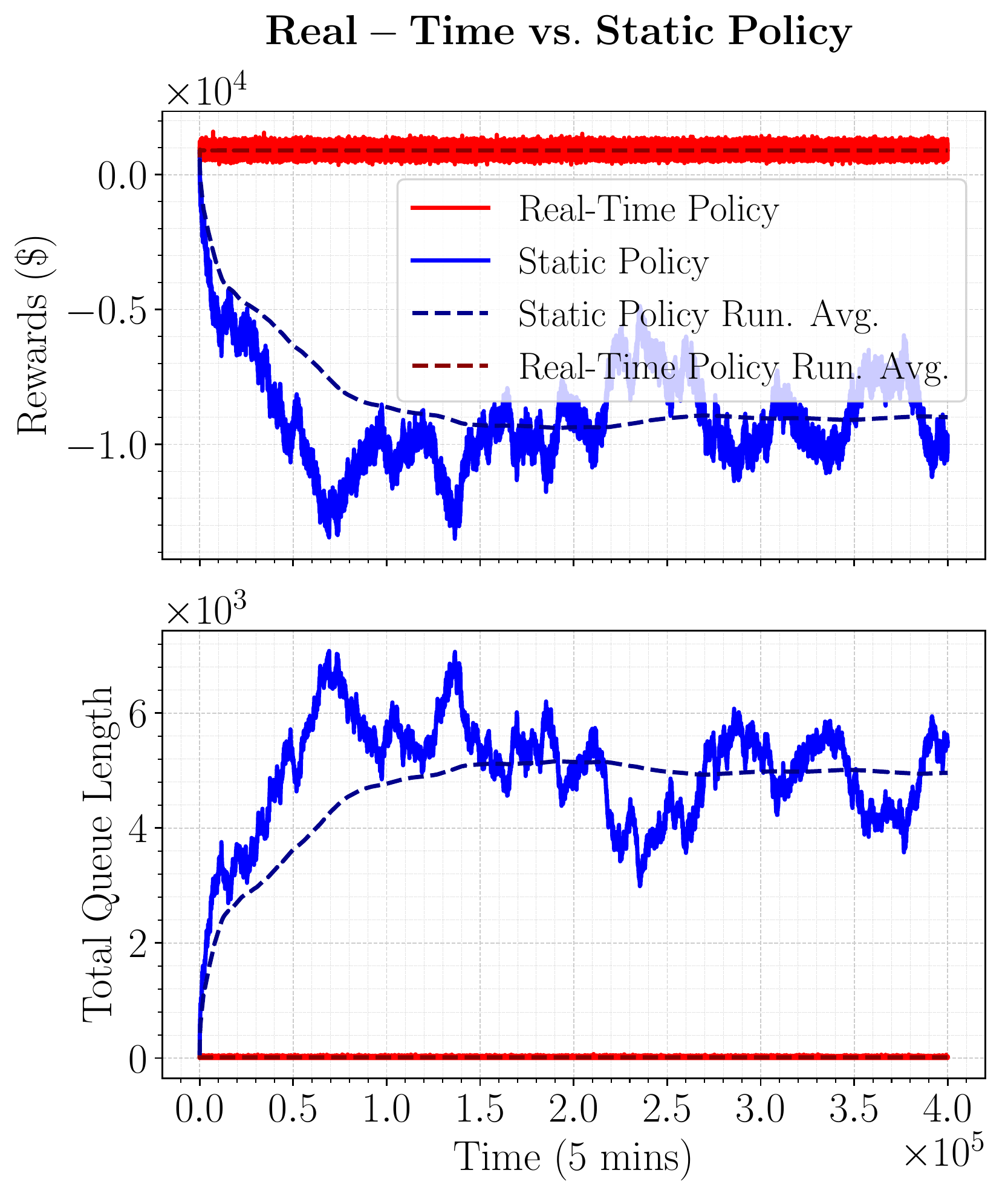}
         \caption{}
         \label{fig:staticvsdynamic_sf}
     \end{subfigure}
     \hfill
     \begin{subfigure}[y]{0.33\textwidth}
         \centering
         \includegraphics[width=\textwidth]{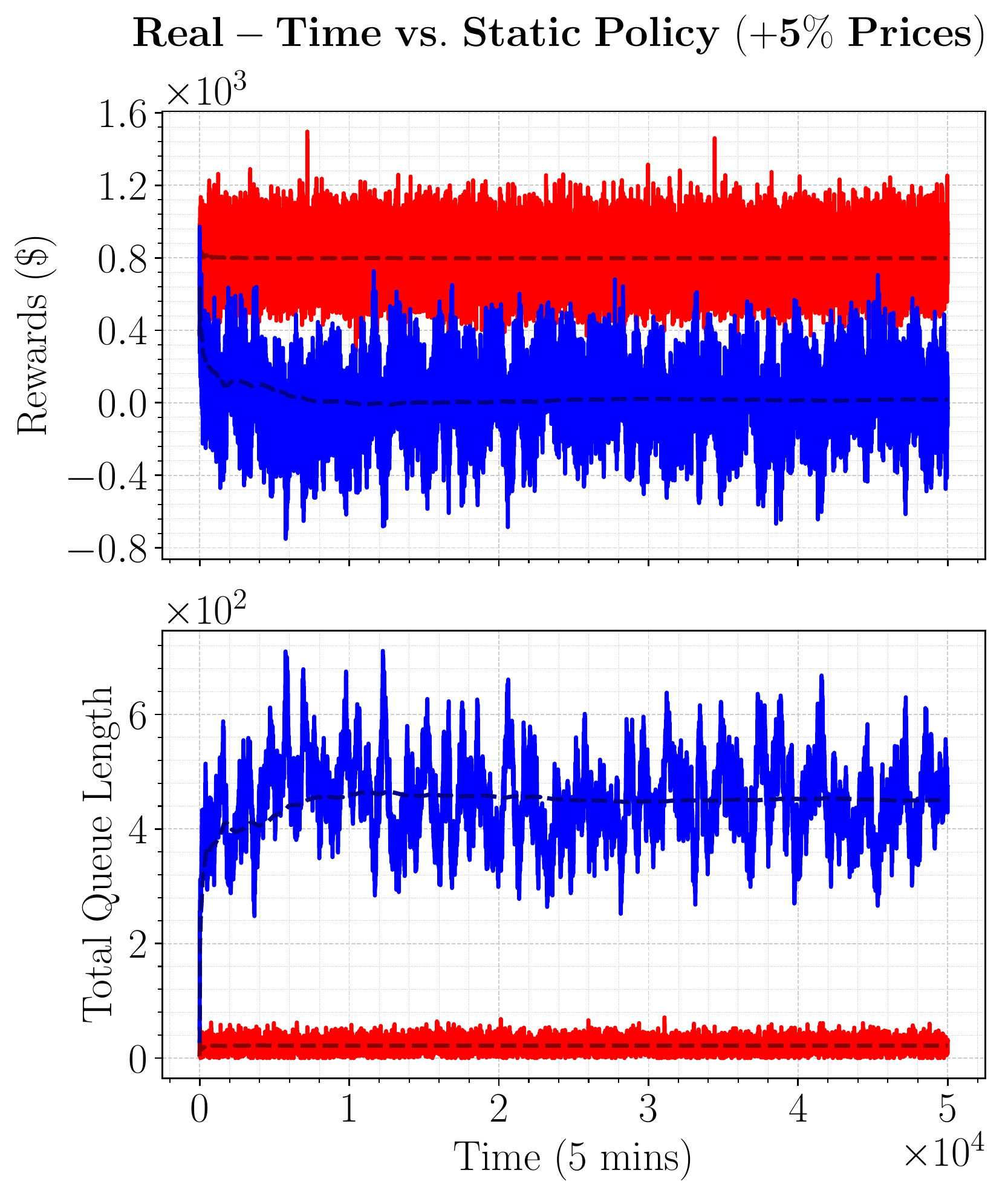}
         \caption{}
         \label{fig:staticvsdynamichigher_sf}
     \end{subfigure}
     \hfill
     \begin{subfigure}[y]{0.33\textwidth}
         \centering
         \includegraphics[width=\textwidth]{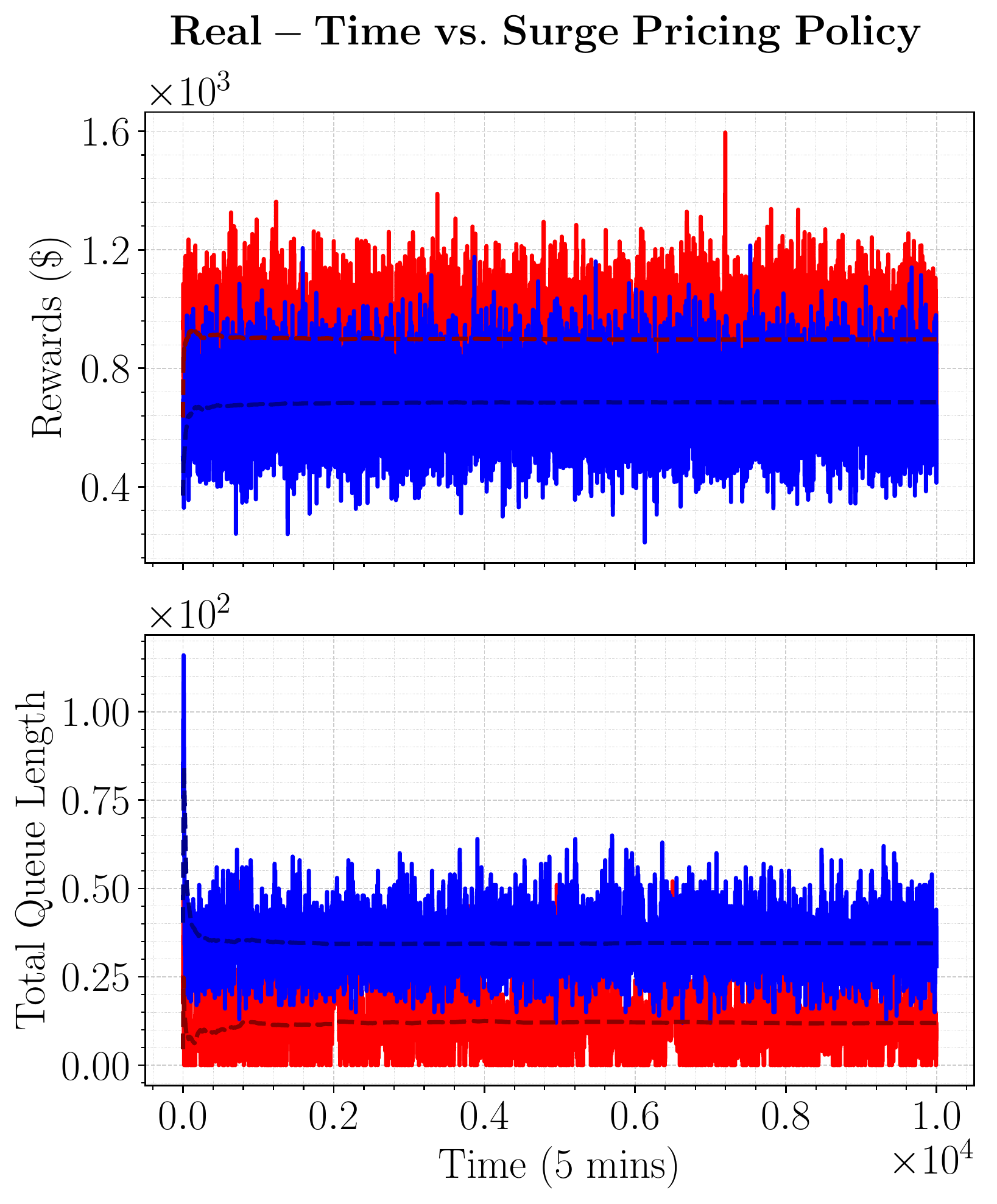}
         \caption{}
         \label{fig:staticvsdynamicsurge_sf}
     \end{subfigure}
        \caption{Comparison of different policies for San Francisco case study. The legends for all figures are the same as the top left figure, where red lines correspond to the real-time policy and blue lines correspond to the static policies. In all scenarios, we use the rewards generated and the total queue length as metrics. In (a), we demonstrate the results from applying the real-time policy and the optimal static policy. In (b), we compare the real-time policy with a sub-optimal static policy, where the prices are $5\%$ higher than the optimal static policy. In (c), we utilize a surge pricing policy along with the optimal static policy and compare with the real-time policy.}
        \label{fig:sfresults}
\end{figure}


\begin{table}[h]
\bt{
 \begin{tabular}{|c||c|c|c|c|c|c|} 
 \hline
\multirow{2}{*}{\backslashbox{Surge Multiplier}{Queue Threshold}} &\multicolumn{2}{c|}{50\%}  
& \multicolumn{2}{c|}{100\%} 
& \multicolumn{2}{c|}{200\%}\\ \cline{2-7}
&Queue&Rewards &Queue&Rewards&Queue&Rewards\\ \hline
1.25&67.62 &\textbf{718.66} &75.92 &715.02 &99.56  &687.45\\\hline
1.5&25.16&650.90 &34.32 &687.71 &49.94 &708.38\\\hline
 2&\textbf{14.06} &331.21 &20.55 &455.25 &44.44 &611.23\\\hline
 \end{tabular}}
\caption{Performances of surge pricing policies for San Francisco case study.}
\label{tab:bayarea_surgepolicies}
\end{table}

\begin{table}[h]
    \bt{
    \begin{tabular}{|c||c|c|c|c|c|}
    \hline
         \backslashbox{Metrics}{\% of opt. static prices}&105\% &110\% &120\% &130\%&140\%  \\\hline
         Average Rewards&4.98 &485.65 & 696.38&\textbf{721.89} &682.76\\\hline
         Average Queue Length& 456.83& 211.04&87.15 & 45.28&\textbf{25.66}\\\hline
    \end{tabular}}
    \caption{Performances of static pricing policies for San Francisco case study.}
    \label{tab:bayarea_higherpricepolicies}
\end{table}

\newpage
\section{Conclusion}\label{sec:conclusion}

In this paper, we developed a real-time control policy based on deep reinforcement learning for operating an AMoD fleet of EVs as well as pricing for rides. Our real-time control policy jointly makes decisions for: 1) vehicle routing in order to serve passenger demand and to rebalance the empty vehicles, 2) vehicle charging in order to sustain energy for rides while exploiting geographical and temporal diversity in electricity prices for cheaper charging options, and 3) pricing for rides in order to adjust the potential demand so that the network is stable and the profits are maximized. 
Furthermore, we formulated the static planning problem associated with the dynamic problem in order to define the optimal static policy for the static planning problem. When implemented correctly, the static policy
provides stability of the queues in the dynamic setting, yet it is not optimal regarding the profits and keeping the queues sufficiently low. Finally, we conducted case studies in Manhattan and San Francisco that demonstrate the performance of our developed policy. \btt{The two case studies on different networks indicate that reinforcement learning can be a universal method for establishing well performing real-time policies that can be applied to many real world networks. Lastly, by doing the Manhattan study with non-electric vehicles and San Francisco study with electric vehicles, we have also demonstrated that a real-time policy using reinforcement learning can be established for both electric and non-electric AMoD systems.}

\bibliographystyle{IEEEtran}
\bibliography{references}

\begin{thebibliography}{10}
\providecommand{\url}[1]{#1}
\csname url@samestyle\endcsname
\providecommand{\newblock}{\relax}
\providecommand{\bibinfo}[2]{#2}
\providecommand{\BIBentrySTDinterwordspacing}{\spaceskip=0pt\relax}
\providecommand{\BIBentryALTinterwordstretchfactor}{4}
\providecommand{\BIBentryALTinterwordspacing}{\spaceskip=\fontdimen2\font plus
\BIBentryALTinterwordstretchfactor\fontdimen3\font minus
  \fontdimen4\font\relax}
\providecommand{\BIBforeignlanguage}[2]{{%
\expandafter\ifx\csname l@#1\endcsname\relax
\typeout{** WARNING: IEEEtran.bst: No hyphenation pattern has been}%
\typeout{** loaded for the language `#1'. Using the pattern for}%
\typeout{** the default language instead.}%
\else
\language=\csname l@#1\endcsname
\fi
#2}}
\providecommand{\BIBdecl}{\relax}
\BIBdecl

\bibitem{companies}
[Online]. Available: https://www.cbinsights.com/research/autonomous-
  driverless-vehicles-corporations-list/.

\bibitem{schulman2017ppo}
J.~Schulman, F.~Wolski, P.~Dhariwal, A.~Radford, and O.~Klimov, ``Proximal
  policy optimization algorithms,'' \emph{arXiv preprint arXiv:1707.06347},
  2017.

\bibitem{queuetheoretical}
R.~Zhang and M.~Pavone, ``{Control of robotic Mobility-on-Demand systems: A
  queueing-theoretical perspective},'' \emph{\textnormal{In} Int. Journal of
  Robotics Research}, vol.~35, no. 1--3, pp. 186--203, 2016.

\bibitem{fluidic}
M.~Pavone, S.~L. Smith, E.~Frazzoli, and D.~Rus, ``{Robotic load balancing for
  Mobility-on-Demand systems},'' \emph{Int. Journal of Robotics Research},
  vol.~31, no.~7, pp. 839--854, 2012.

\bibitem{networkflow}
F.~Rossi, R.~Zhang, Y.~Hindy, and M.~Pavone, ``{Routing autonomous vehicles in
  congested transportation networks: Structural properties and coordination
  algorithms},'' \emph{Autonomous Robots}, vol.~42, no.~7, pp. 1427--1442,
  2018.

\bibitem{markovian}
M.~Volkov, J.~Aslam, and D.~Rus, ``Markov-based redistribution policy model for
  future urban mobility networks,'' \emph{Conference Record - IEEE Conference
  on Intelligent Transportation Systems}, pp. 1906--1911, 09 2012.

\bibitem{wei2019ride}
Q.~Wei, J.~A. Rodriguez, R.~Pedarsani, and S.~Coogan, ``Ride-sharing networks
  with mixed autonomy,'' \emph{arXiv preprint arXiv:1903.07707}, 2019.

\bibitem{zhang_rossi_pavone}
R.~{Zhang}, F.~{Rossi}, and M.~{Pavone}, ``Model predictive control of
  autonomous mobility-on-demand systems,'' in \emph{2016 IEEE International
  Conference on Robotics and Automation (ICRA)}, May 2016.

\bibitem{DBLP:journals/corr/MiaoHLSHZMHP16}
\BIBentryALTinterwordspacing
F.~Miao, S.~Han, S.~Lin, J.~A. Stankovic, H.~Huang, D.~Zhang, S.~Munir, T.~He,
  and G.~J. Pappas, ``Taxi dispatch with real-time sensing data in metropolitan
  areas: {A} receding horizon control approach,'' \emph{CoRR}, vol.
  abs/1603.04418, 2016. [Online]. Available:
  \url{http://arxiv.org/abs/1603.04418}
\BIBentrySTDinterwordspacing

\bibitem{datadrivenmpc}
\BIBentryALTinterwordspacing
R.~Iglesias, F.~Rossi, K.~Wang, D.~Hallac, J.~Leskovec, and M.~Pavone,
  ``Data-driven model predictive control of autonomous mobility-on-demand
  systems,'' \emph{CoRR}, vol. abs/1709.07032, 2017. [Online]. Available:
  \url{http://arxiv.org/abs/1709.07032}
\BIBentrySTDinterwordspacing

\bibitem{mpcmiao}
F.~{Miao}, S.~{Han}, A.~M. {Hendawi}, M.~E. {Khalefa}, J.~A. {Stankovic}, and
  G.~J. {Pappas}, ``Data-driven distributionally robust vehicle balancing using
  dynamic region partitions,'' in \emph{2017 ACM/IEEE 8th International
  Conference on Cyber-Physical Systems (ICCPS)}, April 2017, pp. 261--272.

\bibitem{mpcstochastic}
\BIBentryALTinterwordspacing
M.~Tsao, R.~Iglesias, and M.~Pavone, ``Stochastic model predictive control for
  autonomous mobility on demand,'' \emph{CoRR}, vol. abs/1804.11074, 2018.
  [Online]. Available: \url{http://arxiv.org/abs/1804.11074}
\BIBentrySTDinterwordspacing

\bibitem{fluidtimevarying}
K.~{Spieser}, S.~{Samaranayake}, and E.~{Frazzoli}, ``Vehicle routing for
  shared-mobility systems with time-varying demand,'' in \emph{2016 American
  Control Conference (ACC)}, July 2016, pp. 796--802.

\bibitem{cassandrasarxiv}
R.~M.~A. Swaszek and C.~Cassandras, ``Load balancing in mobility-on-demand
  systems: Reallocation via parametric control using concurrent estimation,''
  \emph{2019 IEEE Intelligent Transportation Systems Conference (ITSC)}, pp.
  2148--2153, 2019.

\bibitem{REPOUX201982}
\BIBentryALTinterwordspacing
M.~Repoux, M.~Kaspi, B.~Boyacı, and N.~Geroliminis, ``Dynamic prediction-based
  relocation policies in one-way station-based carsharing systems with complete
  journey reservations,'' \emph{Transportation Research Part B:
  Methodological}, vol. 130, pp. 82 -- 104, 2019. [Online]. Available:
  \url{http://www.sciencedirect.com/science/article/pii/S019126151930102X}
\BIBentrySTDinterwordspacing

\bibitem{BOYACI2017214}
\BIBentryALTinterwordspacing
B.~Boyacı, K.~G. Zografos, and N.~Geroliminis, ``An integrated
  optimization-simulation framework for vehicle and personnel relocations of
  electric carsharing systems with reservations,'' \emph{Transportation
  Research Part B: Methodological}, vol.~95, pp. 214 -- 237, 2017. [Online].
  Available:
  \url{http://www.sciencedirect.com/science/article/pii/S0191261515301119}
\BIBentrySTDinterwordspacing

\bibitem{WARRINGTON2019110}
\BIBentryALTinterwordspacing
J.~Warrington and D.~Ruchti, ``Two-stage stochastic approximation for dynamic
  rebalancing of shared mobility systems,'' \emph{Transportation Research Part
  C: Emerging Technologies}, vol. 104, pp. 110 -- 134, 2019. [Online].
  Available:
  \url{http://www.sciencedirect.com/science/article/pii/S0968090X18314104}
\BIBentrySTDinterwordspacing

\bibitem{MAO2018179}
\BIBentryALTinterwordspacing
C.~Mao and Z.~Shen, ``A reinforcement learning framework for the adaptive
  routing problem in stochastic time-dependent network,'' \emph{Transportation
  Research Part C: Emerging Technologies}, vol.~93, pp. 179 -- 197, 2018.
  [Online]. Available:
  \url{http://www.sciencedirect.com/science/article/pii/S0968090X18307617}
\BIBentrySTDinterwordspacing

\bibitem{ZHU201430}
\BIBentryALTinterwordspacing
F.~Zhu and S.~V. Ukkusuri, ``Accounting for dynamic speed limit control in a
  stochastic traffic environment: A reinforcement learning approach,''
  \emph{Transportation Research Part C: Emerging Technologies}, vol.~41, pp. 30
  -- 47, 2014. [Online]. Available:
  \url{http://www.sciencedirect.com/science/article/pii/S0968090X1400028X}
\BIBentrySTDinterwordspacing

\bibitem{WALRAVEN2016203}
\BIBentryALTinterwordspacing
E.~Walraven, M.~T. Spaan, and B.~Bakker, ``Traffic flow optimization: A
  reinforcement learning approach,'' \emph{Engineering Applications of
  Artificial Intelligence}, vol.~52, pp. 203 -- 212, 2016. [Online]. Available:
  \url{http://www.sciencedirect.com/science/article/pii/S0952197616000038}
\BIBentrySTDinterwordspacing

\bibitem{ZHU2015487}
\BIBentryALTinterwordspacing
F.~Zhu, H.~A. Aziz, X.~Qian, and S.~V. Ukkusuri, ``A junction-tree based
  learning algorithm to optimize network wide traffic control: A coordinated
  multi-agent framework,'' \emph{Transportation Research Part C: Emerging
  Technologies}, vol.~58, pp. 487 -- 501, 2015, special Issue: Advanced Road
  Traffic Control. [Online]. Available:
  \url{http://www.sciencedirect.com/science/article/pii/S0968090X14003593}
\BIBentrySTDinterwordspacing

\bibitem{lisignal2016}
L.~{Li}, Y.~{Lv}, and F.~{Wang}, ``Traffic signal timing via deep reinforcement
  learning,'' \emph{IEEE/CAA Journal of Automatica Sinica}, vol.~3, no.~3, pp.
  247--254, 2016.

\bibitem{lazar2019learning}
D.~A. Lazar, E.~B{\i}y{\i}k, D.~Sadigh, and R.~Pedarsani, ``Learning how to
  dynamically route autonomous vehicles on shared roads,'' \emph{arXiv preprint
  arXiv:1909.03664}, 2019.

\bibitem{rldecentralized1}
M.~Han, P.~Senellart, S.~Bressan, and H.~Wu, ``Routing an autonomous taxi with
  reinforcement learning,'' in \emph{CIKM}, 2016.

\bibitem{rldecentralized2}
M.~{Guériau} and I.~{Dusparic}, ``Samod: Shared autonomous mobility-on-demand
  using decentralized reinforcement learning,'' in \emph{2018 21st
  International Conference on Intelligent Transportation Systems (ITSC)}, Nov
  2018, pp. 1558--1563.

\bibitem{rldecentralized3}
J.~{Wen}, J.~{Zhao}, and P.~{Jaillet}, ``Rebalancing shared mobility-on-demand
  systems: A reinforcement learning approach,'' in \emph{2017 IEEE 20th
  International Conference on Intelligent Transportation Systems (ITSC)}, Oct
  2017, pp. 220--225.

\bibitem{linrldecentralized}
\BIBentryALTinterwordspacing
K.~Lin, R.~Zhao, Z.~Xu, and J.~Zhou, ``Efficient large-scale fleet management
  via multi-agent deep reinforcement learning,'' in \emph{Proceedings of the
  24th ACM SIGKDD International Conference on Knowledge Discovery \& Data
  Mining}, ser. KDD ’18.\hskip 1em plus 0.5em minus 0.4em\relax New York, NY,
  USA: Association for Computing Machinery, 2018, p. 1774–1783. [Online].
  Available: \url{https://doi.org/10.1145/3219819.3219993}
\BIBentrySTDinterwordspacing

\bibitem{MAO2020102626}
\BIBentryALTinterwordspacing
C.~Mao, Y.~Liu, and Z.-J.~M. Shen, ``Dispatch of autonomous vehicles for taxi
  services: A deep reinforcement learning approach,'' \emph{Transportation
  Research Part C: Emerging Technologies}, vol. 115, p. 102626, 2020. [Online].
  Available:
  \url{http://www.sciencedirect.com/science/article/pii/S0968090X19312227}
\BIBentrySTDinterwordspacing

\bibitem{New_TSG_SmartEVGrid}
E.~Veldman and R.~A. Verzijlbergh, ``Distribution grid impacts of smart
  electric vehicle charging from different perspectives,'' \emph{IEEE
  Transactions on Smart Grid}, vol.~6, no.~1, pp. 333--342, Jan 2015.

\bibitem{New_Survey}
W.~Su, H.~Eichi, W.~Zeng, and M.~Chow, ``A survey on the electrification of
  transportation in a smart grid environment,'' \emph{IEEE Transactions on
  Industrial Informatics}, vol.~8, no.~1, pp. 1--10, Feb 2012.

\bibitem{New_20}
J.~C. Mukherjee and A.~Gupta, ``A review of charge scheduling of electric
  vehicles in smart grid,'' \emph{IEEE Systems Journal}, vol.~9, no.~4, pp.
  1541--1553, Dec 2015.

\bibitem{agentbased}
T.~D. Chen, K.~M. Kockelman, and J.~P. Hanna, ``{Operations of a Shared,
  Autonomous, Electric Vehicle Fleet: Implications of Vehicle \& Charging
  Infrastructure Decisions},'' \emph{Transportation Research Part A: Policy and
  and Practice}, vol.~94, pp. 243--254, 2016.

\bibitem{BONGIOVANNI2019436}
\BIBentryALTinterwordspacing
C.~Bongiovanni, M.~Kaspi, and N.~Geroliminis, ``The electric autonomous
  dial-a-ride problem,'' \emph{Transportation Research Part B: Methodological},
  vol. 122, pp. 436 -- 456, 2019. [Online]. Available:
  \url{http://www.sciencedirect.com/science/article/pii/S0191261517309669}
\BIBentrySTDinterwordspacing

\bibitem{nate}
N.~Tucker, B.~Turan, and M.~Alizadeh, ``{Online Charge Scheduling for Electric
  Vehicles in Autonomous Mobility on Demand Fleets},'' \emph{\textnormal{In}
  Proc. IEEE Int. Conf. on Intelligent Transportation Systems}, 2019.

\bibitem{berkay}
\BIBentryALTinterwordspacing
B.~Turan, N.~Tucker, and M.~Alizadeh, ``{Smart Charging Benefits in Autonomous
  Mobility on Demand Systems},'' \emph{\textnormal{In} Proc. IEEE Int. Conf. on
  Intelligent Transportation Systems}, 2019. [Online]. Available:
  \url{https://arxiv.org/abs/1907.00106}
\BIBentrySTDinterwordspacing

\bibitem{rossi_iglesias_alizadeh}
F.~Rossi, R.~Iglesias, M.~Alizadeh, and M.~Pavone, ``On the interaction between
  autonomous mobility-on-demand systems and the power network: models and
  coordination algorithms,'' \emph{Robotics: Science and Systems XIV}, Jun
  2018.

\bibitem{chen_kockelman}
T.~D. Chen and K.~M. Kockelman, ``Management of a shared autonomous electric
  vehicle fleet: Implications of pricing schemes,'' \emph{Transportation
  Research Record}, vol. 2572, no.~1, pp. 37--46, 2016.

\bibitem{pricing2}
\BIBentryALTinterwordspacing
Y.~Guan, A.~M. Annaswamy, and H.~E. Tseng, ``Cumulative prospect theory based
  dynamic pricing for shared mobility on demand services,'' \emph{CoRR}, vol.
  abs/1904.04824, 2019. [Online]. Available:
  \url{http://arxiv.org/abs/1904.04824}
\BIBentrySTDinterwordspacing

\bibitem{jointopt}
\BIBentryALTinterwordspacing
C.~J.~R. Sheppard, G.~S. Bauer, B.~F. Gerke, J.~B. Greenblatt, A.~T. Jenn, and
  A.~R. Gopal, ``Joint optimization scheme for the planning and operations of
  shared autonomous electric vehicle fleets serving mobility on demand,''
  \emph{Transportation Research Record}, vol. 2673, no.~6, pp. 579--597, 2019.
  [Online]. Available: \url{https://doi.org/10.1177/0361198119838270}
\BIBentrySTDinterwordspacing

\bibitem{bimpikis}
K.~Bimpikis, O.~Candogan, and D.~Sab{\'a}n, ``Spatial pricing in ride-sharing
  networks,'' \emph{Operations Research}, vol.~67, pp. 744--769, 2019.

\bibitem{pedarsani2017robust}
R.~Pedarsani, J.~Walrand, and Y.~Zhong, ``Robust scheduling for flexible
  processing networks,'' \emph{Advances in Applied Probability}, vol.~49,
  no.~2, pp. 603--628, 2017.

\bibitem{kaelbling1996reinforcement}
L.~P. Kaelbling, M.~L. Littman, and A.~W. Moore, ``Reinforcement learning: A
  survey,'' \emph{Journal of artificial intelligence research}, vol.~4, pp.
  237--285, 1996.

\bibitem{SUMO2018}
\BIBentryALTinterwordspacing
P.~A. Lopez, M.~Behrisch, L.~Bieker-Walz, J.~Erdmann, Y.-P. Fl{\"o}tter{\"o}d,
  R.~Hilbrich, L.~L{\"u}cken, J.~Rummel, P.~Wagner, and E.~Wie{\ss}ner,
  ``Microscopic traffic simulation using sumo,'' in \emph{The 21st IEEE
  International Conference on Intelligent Transportation Systems}.\hskip 1em
  plus 0.5em minus 0.4em\relax IEEE, 2018. [Online]. Available:
  \url{<https://elib.dlr.de/124092/}
\BIBentrySTDinterwordspacing

\bibitem{surveyrl}
I.~{Grondman}, L.~{Busoniu}, G.~A.~D. {Lopes}, and R.~{Babuska}, ``A survey of
  actor-critic reinforcement learning: Standard and natural policy gradients,''
  \emph{IEEE Transactions on Systems, Man, and Cybernetics, Part C
  (Applications and Reviews)}, vol.~42, no.~6, pp. 1291--1307, Nov 2012.

\bibitem{qlearning1}
\BIBentryALTinterwordspacing
C.~J. C.~H. Watkins and P.~Dayan, ``Q-learning,'' \emph{Machine Learning},
  vol.~8, no.~3, pp. 279--292, May 1992. [Online]. Available:
  \url{https://doi.org/10.1007/BF00992698}
\BIBentrySTDinterwordspacing

\bibitem{sarsa}
G.~A. Rummery and M.~Niranjan, ``On-line q-learning using connectionist
  systems,'' Tech. Rep., 1994.

\bibitem{reinforce}
\BIBentryALTinterwordspacing
R.~J. Williams, ``Simple statistical gradient-following algorithms for
  connectionist reinforcement learning,'' \emph{Machine Learning}, vol.~8,
  no.~3, pp. 229--256, May 1992. [Online]. Available:
  \url{https://doi.org/10.1007/BF00992696}
\BIBentrySTDinterwordspacing

\bibitem{actorcritic1}
A.~G. {Barto}, R.~S. {Sutton}, and C.~W. {Anderson}, ``Neuronlike adaptive
  elements that can solve difficult learning control problems,'' \emph{IEEE
  Transactions on Systems, Man, and Cybernetics}, vol. SMC-13, no.~5, pp.
  834--846, Sep. 1983.

\bibitem{actorcritic2}
I.~H. Witten, ``An adaptive optimal controller for discrete-time markov
  environments,'' \emph{Information and Control}, vol.~34, pp. 286--295, 1977.

\bibitem{schulman2015trpo}
\BIBentryALTinterwordspacing
J.~Schulman, S.~Levine, P.~Moritz, M.~I. Jordan, and P.~Abbeel, ``Trust region
  policy optimization,'' \emph{CoRR}, vol. abs/1502.05477, 2015. [Online].
  Available: \url{http://arxiv.org/abs/1502.05477}
\BIBentrySTDinterwordspacing

\bibitem{cvx}
M.~Grant and S.~Boyd, ``{CVX}: Matlab software for disciplined convex
  programming, version 2.1,'' \url{http://cvxr.com/cvx}, Mar. 2014.

\bibitem{1606.01540}
G.~Brockman, V.~Cheung, L.~Pettersson, J.~Schneider, J.~Schulman, J.~Tang, and
  W.~Zaremba, ``Openai gym,'' 2016.

\bibitem{stable-baselines}
A.~Hill, A.~Raffin, M.~Ernestus, A.~Gleave, R.~Traore, P.~Dhariwal, C.~Hesse,
  O.~Klimov, A.~Nichol, M.~Plappert, A.~Radford, J.~Schulman, S.~Sidor, and
  Y.~Wu, ``Stable baselines,''
  \url{https://github.com/hill-a/stable-baselines}, 2018.

\bibitem{avgevprice}
{The average electric car in the US is getting cheaper}. [Online]. Available:
  https://qz.com/1695602/the-average-electric-vehicle-is-getting-cheaper-in-the-us/.

\bibitem{electricityprices}
\BIBentryALTinterwordspacing
 [Online]. Available: \url{http://oasis.caiso.com}
\BIBentrySTDinterwordspacing

\bibitem{avgwages}
{United States Average Hourly Wages}. [Online]. Available:
  https://tradingeconomics.com/united-states/wages.

\bibitem{avgcostofdriving}
{How much does driving your car cost, per minute?} [Online]. Available:
  https://www.bostonglobe.com/ideas/2014/08/08/how-much-driving-really-costs-per-minute/BqnNd2q7jETedLhxxzY2CI/story.html.

\bibitem{manhattantaxidata}
\BIBentryALTinterwordspacing
 [Online]. Available:
  \url{https://www1.nyc.gov/site/tlc/about/tlc-trip-record-data.page}
\BIBentrySTDinterwordspacing

\bibitem{lillicrap2015ddpg}
T.~P. Lillicrap, J.~J. Hunt, A.~Pritzel, N.~Heess, T.~Erez, Y.~Tassa,
  D.~Silver, and D.~Wierstra, ``Continuous control with deep reinforcement
  learning,'' \emph{arXiv preprint arXiv:1509.02971}, 2015.

\bibitem{sfmap}
\BIBentryALTinterwordspacing
 [Online]. Available: \url{http://tncstoday.sfcta.org/}
\BIBentrySTDinterwordspacing

\bibitem{epfl-mobility-20090224}
M.~Piorkowski, N.~Sarafijanovic-Djukic, and M.~Grossglauser, ``{CRAWDAD}
  dataset epfl/mobility (v. 2009-02-24),'' Downloaded from
  \url{https://crawdad.org/epfl/mobility/20090224}, Feb. 2009.

\bibitem{fluidstability}
J.~G. Dai, ``On positive harris recurrence of multiclass queueing networks: A
  unified approach via fluid limit models,'' \emph{Annals of Applied
  Probability}, vol.~5, pp. 49--77, 1995.

\end{thebibliography}
\newpage
\begin{appendices}

\section{Proof of Proposition \ref{prop:stability}}\label{sec:proofpropstability}
To prove Proposition \ref{prop:stability}, we first formulate the static optimization problem via a network flow model that characterizes the \textit{capacity region} of the network for a given set of prices $\ell_{ij}(t)=\ell_{ij}\; \forall t$ (Hence, $\Lambda_{ij}(t)=\Lambda_{ij} \; \forall t$). 
The capacity region is defined as the set of all arrival rates $[\Lambda_{ij}]_{i,j \in \mathcal M}$, where there exists a charging and routing policy under which the queueing network of the system is stable.
Let $x_i^v$ be the number of vehicles available at node $i$, $\alpha_{ij}^v$ be the fraction of vehicles at node $i$ with energy level $v$ being routed to node $j$, and $\alpha_{ic}^v$ be the fraction of vehicles charging at node $i$ starting with energy level $v$. We say the static vehicle allocation   for node $i$ and energy level $v$ is feasible if $\alpha_{ic}^v + \sum_{\substack{j\in{\cal M}\\ j\neq i}} \alpha_{ij}^v \leq 1$.

The optimization problem that characterizes the capacity region of the network ensures that the total number of vehicles routed from $i$ to $j$ is at least as large as the nominal arrival rate to the queue $(i,j)$. Namely, the vehicle allocation problem can be formulated as follows:
\allowdisplaybreaks[3]
\begin{subequations}
\label{eq:capacityregionproblem}
\begin{align}
\label{eq:capobj}&\underset{x_i^v,\alpha_{ij}^v,\alpha_{ic}^v}{\text{min}}& & \rho\\
& \text{subject to}
\label{eq:capconst1}& & \Lambda_{ij} \leq \sum_{v=v_{ij}}^{v_{\max}}x_i^v\alpha_{ij}^v \quad\forall i,j\in \mathcal M,\\
\label{eq:capconst2}& & & \rho \geq \alpha_{ic}^v + \sum_{\substack{j\in{\cal M}\\ j\neq i}} \alpha_{ij}^v\quad \forall i \in \mathcal M,\; \forall v \in \mathcal V,\\
\label{eq:capconst3} & & & x_i^v = x_i^{v-1}\alpha_{ic}^{v-1}+\sum_{j\in{\cal M}} x_i^{v+v_{ji}}\alpha_{ji}^{v+v_{ji}}\quad \forall i \in \mathcal M,\; \forall v \in \mathcal V,\\
\label{eq:capconst5}& & & \alpha_{ic}^{v_{\max}}=0\quad \forall i\in \mathcal M,\\
\label{eq:capconst6}& & & \alpha_{ij}^v=0\quad \forall v<v_{ij},\; \forall  i,j\in \mathcal M\\
& & & x_i^v\geq 0, \; \alpha_{ij}^v\geq 0 \; \alpha_{ic}^v\geq 0, ~\forall i,j\in \mathcal M,\;\forall v\in\mathcal V,\\
& & & x_i^v=\alpha_{ic}^v=\alpha_{ij}^v=0\quad \forall v\notin\mathcal V,\; \forall  i,j\in \mathcal M.
\end{align}
\end{subequations}

The constraint \eqref{eq:capconst2} upper bounds the allocation of vehicles for each node $i$ and energy level $v$. The constraints \eqref{eq:capconst3}-\eqref{eq:capconst6} are similar to those of optimization problem \eqref{eq:flowoptimization} with $x_i^v=x_{ic}^v+\sum_{j\in{\cal M}}x_{ij}^v$, $\alpha_{ic}^v=x_{ic}^v/x_i^v$, and $\alpha_{ij}^v=x_{ij}^v/x_i^v$.
\begin{lemma}\label{lem:capacityregion}
Let the optimal value of \eqref{eq:capacityregionproblem} be $\rho^*$. Then, $\rho^*\leq 1$ is a necessary and sufficient condition of rate stability of the system under some routing and charging policy.
\end{lemma}
\begin{proof}
Consider the fluid scaling of the queueing network, $Q_{ij}^{rt}=\frac{q_{ij}(\lfloor rt\rfloor)}{r}$ (see \cite{fluidstability} for more discussion on the stability of fluid models), and let $Q_{ij}^t$ be the corresponding fluid limit. The fluid model dynamics is as follows:
\begin{equation*}
    \label{eq:fluiddynamics}
    Q_{ij}^t=Q_{ij}^0+A_{ij}^t-X_{ij}^t,
\end{equation*}
where $A_{ij}^t$ is the total number of riders from node $i$ to node $j$ that have arrived to the network until time $t$ and $X_{ij}^t$ is the total number of vehicles routed from node $i$ to $j$ up to time $t$. Suppose that $\rho^* > 1$ and there exists a policy under which for all $t\geq 0$ and for all origin-destination pairs $(i,j)$, $Q_{ij}^t=0$. Pick a point $t_1$, where $Q_{ij}^{t_1}$ is differentiable for all $(i,j)$. Then, for all $(i,j)$, $\dot{Q}_{ij}^{t_1}=0$. Since $\dot{A}_{ij}^{t_1}=\Lambda_{ij}$, this implies $\dot{X}_{ij}^{t_1}=\Lambda_{ij}$. On the other hand, $\dot{X}_{ij}^{t_1}$ is the total number of vehicles routed from $i$ to $j$ at $t_1$. This implies $\Lambda_{ij}=\sum_{v=v_{ij}}^{v_{\max}}x_i^v\alpha_{ij}^v$ for all $(i,j)$ and
there exists $\alpha_{ij}^v$ and $\alpha_{ic}^v$ at time $t_1$ such that the flow balance constraints hold and the allocation vector $[\alpha_{ij}^v\; \alpha_{ic}^v]$ is feasible, i.e. $\alpha_{ic}^v+\sum_{\substack{j=1\\j\neq i}}^m\alpha_{ij}^v\leq 1$. This contradicts $\rho^*>1$.

Now suppose $\rho^*\leq 1$ and $\alpha^*=[\alpha_{ij}^{v*}\; \alpha_{ic}^{v*}]$ is an allocation vector that solves the static problem. The cumulative number of vehicles routed from node $i$ to $j$ up to time $t$ is $S_{ij}^t=\sum_{v=v_{ij}}^{v_{\max}}x_i^v\alpha_{ij}^v t=\sum_{v=0}^{v_{\max}}x_i^v\alpha_{ij}^v t\geq \Lambda_{ij}t$. Suppose that for some  origin-destination pair $(i,j)$, the queue $Q_{ij}^{t_1}\geq \epsilon >0$ for some positive $t_1$ and $\epsilon$. By continuity of the fluid limit, there exists $t_0 \in (0,t_1)$ such that $Q_{ij}^{t_0}=\epsilon/2$ and $Q_{ij}^t>0$ for $t \in [t_0,t_1]$. Then, $\dot{Q}_{ij}^t > 0$ implies $\Lambda_{ij}>\sum_{v=v_{ij}}^{v_{\max}}x_i^v\alpha_{ij}^v$, which is a contradiction.
\end{proof}
 By Lemma \ref{lem:capacityregion}, the \textit{capacity region} $C_\Lambda$ of the network is the set of all $\Lambda_{ij} \in \mathbb{R}^+$ for which the corresponding optimal solution to the optimization problem \eqref{eq:capacityregionproblem} satisfies $\rho^*\leq 1$. As long as $\rho^*\leq 1$, there exists a routing and charging policy such that the queues will be bounded away from infinity.

The platform operator's goal is to maximize its profits by setting prices and making routing and charging decisions such that the system remains stable. In its most general form, the problem can be formulated as follows:
\begin{equation}
\label{eq:capacitytoprofit}
\begin{aligned}
&\underset{\ell_{ij},x_i^v,\alpha_{ij}^v,\alpha_{ic}^v}{\text{max}}& & U(\Lambda_{ij}(\ell_{ij}),x_i^v,\alpha_{ij}^v,\alpha_{ic}^v)\\
& \text{subject to}
& & [\Lambda_{ij}(\ell_{ij})]_{i,j \in \mathcal{M}}\in C_\Lambda,\\
\end{aligned}
\end{equation}
where $U(\cdot)$ is the utility function that depends on the prices, demand for rides and the vehicle decisions. 

Recall that $x_{ic}^v=x_{i}^v\alpha_{ic}^v$ and $x_{ij}^v=x_{i}^v\alpha_{ij}^v$. Using these variables and noting that $\alpha_{ic}^v+\sum_{j\in{\cal M}} \alpha_{ij}^v=1$ when $\rho^*\leq 1$, the platform operator's profit maximization problem can be stated as \eqref{eq:flowoptimization}. A feasible solution of \eqref{eq:flowoptimization} guarantees rate stability of the system, since the corresponding vehicle allocation problem \eqref{eq:capacityregionproblem} has solution $\rho^*\leq 1$.

\section{Proof of Proposition \ref{prop:marginalprices}}\label{sec:proofprop2}
For brevity of notation, let $\beta+p_i=P_i$. Let $\nu_{ij}$ be the dual variables corresponding to the demand satisfaction constraints and $\mu^v_i$ be the dual variables corresponding to the flow balance constraints.
Since the optimization problem   \eqref{eq:flowoptimization} is a convex quadratic maximization problem (given a with uniform $F(\cdot)$) and Slater's condition is satisfied, strong duality holds. We can write the dual problem as:
\begin{subequations}\label{eq:dualproblem}
\begin{align}
\label{eq:dualobjective}
    &\underset{\nu_{ij},\mu_i^v}{\text{min}}\underset{\ell_{ij}}{\text{max}}
    & &\sum_{i=1}^m\sum_{j=1}^m\left(\lambda_{ij}(1-\frac{\ell_{ij}}{\ell_{\max}})\left(\ell_i-\nu_{ij}\right)\right)\\
    & \text{subject to}
    & & \nu_{ij}\geq 0,\\
   & & & \nu_{ij}+\mu_i^v-\mu^{v-v_{ij}}-\beta \tau_{ij}\leq 0,\\
    & & & \mu_i^v-\mu_i^{v+1}-P_i \leq  0\quad \forall i,j,v.
\end{align}
\end{subequations}
For fixed $\nu_{ij}$
and $\mu_i^v$, the inner maximization results in the optimal prices:
\begin{equation}
\label{eq:optimalpricesproof}
    \ell_{ij}^*=\frac{\ell_{\max}+\nu_{ij}}{2}.
\end{equation}
By strong duality, the optimal primal solution satisfies the dual solution with optimal dual variables $\nu_{ij}^*$
and ${\mu_i^v}^*$, which completes the first part of the proposition. The dual problem with optimal prices in \eqref{eq:optimalpricesproof} can be written as:
\begin{subequations}
    \label{eq:dualwithoptimalprices}
    \begin{align}\label{eq:dualobjectivewithoptimalprices}
    &\underset{\nu_{ij}, \mu_i^v}{\text{min}}
    & &\sum_{i=1}^m\sum_{j=1}^m\frac{\lambda_{ij}}{\ell_{\max}}\left(\frac{\ell_{\max}-\nu_{ij}}{2}\right)^2\\
    & \text{subject to}
\label{eq:dualconstraint1}    & & \nu_{ij}\geq 0,\\
 \label{eq:dualconstraint3}  & & & \nu_{ij}+\mu_i^v-\mu_j^{v-v_{ij}}-\beta \tau_{ij}\leq 0,\\
\label{eq:dualconstraint4}    & & & \mu_i^v-\mu_i^{v+1}-P_i \leq  0\quad \forall i,j,v.
    \end{align}
\end{subequations}
The objective function in \eqref{eq:dualobjectivewithoptimalprices} with optimal dual variables, along with \eqref{eq:optimalpricesproof} suggests:
\begin{equation*}
    P=\sum_{i=1}^m\sum_{j=1}^m\frac{\lambda_{ij}}{\ell_{\max}}(\ell_{\max}-\ell_{ij}^*)^2,
\end{equation*}
where profits $P$ is the value of the objective function of both optimal and dual problems. To get the upper bound on prices, we go through the following algebraic calculations using the constraints. The inequality \eqref{eq:dualconstraint4} gives:
\begin{equation}
    \label{eq:dualconstraint5}
        \mu_i^{v-v_{ji}}\leq v_{ji}P_i+\mu_i^v,
\end{equation}
and equivalently:
\begin{equation}
    \label{eq:dualconstraint6}
    \mu_j^{v-v_{ij}}\leq v_{ij}P_j+\mu_j^v.
\end{equation}
The inequalities \eqref{eq:dualconstraint3} and \eqref{eq:dualconstraint1} yield:
\begin{equation*}
    \label{eq:dualconstraint7}
    \mu_i^v-\mu_j^{v-v_{ij}}-\beta \tau_{ij}\leq 0,
\end{equation*}
and equivalently:
\begin{equation}
    \label{eq:dualconstraint8}
    \mu_j^v-\mu_i^{v-v_{ji}}-\beta \tau_{ji}\leq 0,
\end{equation}
Inequalities \eqref{eq:dualconstraint5} and \eqref{eq:dualconstraint8}:
\begin{equation}
\label{eq:dualconstraint9}
\mu_j^v\leq \mu_i^v+\beta\tau_{ji}+v_{ji}P_i.
\end{equation}
And finally, the constraint \eqref{eq:dualconstraint3}:
\begin{align*}
    \nu_{ij}&\leq\beta \tau_{ij}+\mu_j^{v-v_{ij}}-\mu_i^v\\
    &\overset{\eqref{eq:dualconstraint6}}{\leq}\beta \tau_{ij}+v_{ij}P_j+\mu_j^v-\mu_i^v\\
    &\overset{\eqref{eq:dualconstraint9}}{\leq}\beta
    \tau_{ij}+v_{ij}P_j+\beta \tau_{ji}+v_{ji}P_i.
\end{align*}
Replacing $P_i=p_i+\beta$ and rearranging the terms:
\begin{align}
    \label{eq:lambdaupperbound}
    \nu_{ij}\leq \beta(\tau_{ij}+\tau_{ji}+v_{ij}+v_{ji})+v_{ij}p_j+v_{ji}p_i.
\end{align}
Using the upper bound on the dual variables $\nu_{ij}$ and \eqref{eq:optimalpricesproof}, we can upper bound the optimal prices.
\end{appendices}
\end{document}